\documentclass{article}

\usepackage{microtype}
\usepackage{graphicx}
\usepackage{subfigure}
\usepackage{booktabs} 

\usepackage{hyperref}

\usepackage{lipsum}
\usepackage{wrapfig}
\usepackage{amsmath}
\usepackage{amsthm}
\usepackage{amssymb}
\newcommand{\bfx}{\mathbf{x}}
\newcommand{\bfz}{\mathbf{z}}
\DeclareMathOperator{\FAPE}{FAPE}
\DeclareMathOperator{\BetaCarbons}{BetaCarbons}
\DeclareMathOperator{\Uniform}{Uniform}
\DeclareMathOperator{\RMSDAlign}{RMSDAlign}
\DeclareMathOperator{\HarmonicPrior}{HarmonicPrior}
\DeclareMathOperator{\AlphaFold}{AlphaFold}
\newcommand{\pluseq}{\mathrel{+}=}
\DeclareMathOperator{\Bin}{Bin}
\DeclareMathOperator{\Tr}{Tr}
\DeclareMathOperator{\Linear}{Linear}
\DeclareMathOperator{\RMWD}{RMWD}
\DeclareMathOperator{\GaussianFourierEmbedding}{GaussianFourierEmbedding}
\DeclareMathOperator{\length}{length}
\DeclareMathOperator{\OneHot}{OneHot}

\DeclareMathOperator{\TriangleAttentionStartingNode}{TriangleAttentionStartingNode}
\DeclareMathOperator{\TriangleAttentionEndingNode}{TriangleAttentionEndingNode}
\DeclareMathOperator{\TriangleMultiplicationOutgoing}{TriangleMultiplicationOutgoing}
\DeclareMathOperator{\TriangleMultiplicationIncoming}{TriangleMultiplicationIncoming}
\DeclareMathOperator{\PairTransition}{PairTransition}
\usepackage{multirow,makecell}

\usepackage[accepted]{icml2024}

\usepackage{amsmath}
\usepackage{amssymb}
\usepackage{mathtools}
\usepackage{amsthm}

\usepackage[capitalize,noabbrev]{cleveref}

\theoremstyle{plain}
\newtheorem{theorem}{Theorem}[section]
\newtheorem{proposition}[theorem]{Proposition}

\theoremstyle{definition}

\theoremstyle{remark}

\usepackage[textsize=tiny]{todonotes}

\icmltitlerunning{AlphaFold Meets Flow Matching for Generating Protein Ensembles}
\begin{document}

\twocolumn[
\icmltitle{AlphaFold Meets Flow Matching for Generating Protein Ensembles}



\icmlsetsymbol{equal}{*}

\begin{icmlauthorlist}
\icmlauthor{Bowen Jing}{csail}
\icmlauthor{Bonnie Berger}{csail,math}
\icmlauthor{Tommi Jaakkola}{csail}
\end{icmlauthorlist}

\icmlaffiliation{csail}{CSAIL, Massachusetts Institute of Technology}
\icmlaffiliation{math}{Department of Mathematics, Massachusetts Institute of Technology}

\icmlcorrespondingauthor{Bowen Jing}{bjing@mit.edu}

\icmlkeywords{Machine Learning, ICML}

\vskip 0.3in
]



\printAffiliationsAndNotice{}
\begin{abstract}
The biological functions of proteins often depend on dynamic structural ensembles. In this work, we develop a flow-based generative modeling approach for learning and sampling the conformational landscapes of proteins. We repurpose highly accurate single-state predictors such as AlphaFold and ESMFold and fine-tune them under a custom flow matching framework to obtain sequence-conditioned \emph{generative} models of protein structure called Alpha\textsc{Flow} and ESM\textsc{Flow}. When trained and evaluated on the PDB, our method provides a superior combination of precision and diversity compared to AlphaFold with MSA subsampling. When further trained on ensembles from all-atom MD, our method accurately captures conformational flexibility, positional distributions, and higher-order ensemble observables for unseen proteins. Moreover, our method can diversify a static PDB structure with faster wall-clock convergence to certain equilibrium properties than replicate MD trajectories, demonstrating its potential as a proxy for expensive physics-based simulations. Code is available at \url{https://github.com/bjing2016/alphaflow}. 
\end{abstract}
\section{Introduction}
Proteins adopt complex three-dimensional structures, often as members of structural \emph{ensembles} with distinct states, collective motions, and disordered fluctuations, to carry out their biological functions. For example, conformational changes are critical in the function of transporters, channels, and enzymes, and the properties of equilibrium ensembles help govern the strength and selectivity of molecular interactions \citep{meller2023predicting, vogele2023functional}. 
While deep learning methods such as AlphaFold \citep{jumper2021highly} have excelled in the single-state modeling of experimental protein structures, they fail to account for this conformational heterogeneity \citep{lane2023protein, ourmazd2022structural}.
Hence, a method which builds upon the level of accuracy of single-structure predictors, but reveals underlying structural ensembles, would be of great value to structural biologists.

Existing machine learning approaches for generating structural ensembles have focused on inference-time interventions in AlphaFold that modify the multiple sequence alignment (MSA) input \citep{del2022sampling, stein2022speach_af, wayment2023predicting}, resulting in a different structure prediction for each version of the MSA. While these approaches have demonstrated some success, they suffer from two key limitations. First, by operating on the MSA, they cannot be generalized to structure predictors based on protein language models (PLMs) such as ESMFold \citep{lin2023evolutionary} or OmegaFold \citep{wu2022high}, which have grown in popularity due to their fast runtime and ease of use. Secondly, these inference-time interventions do not provide the capability to train on protein ensembles from beyond the PDB---for example, ensembles from molecular dynamics, which are of significant scientific interest but can be extremely expensive to simulate \citep{shaw2010atomic}.

To address these limitations, in this work we combine AlphaFold and ESMFold with \emph{flow matching}, a recent generative modeling framework \citep{lipman2022flow, albergo2022building}, to propose a principled method for sampling the conformational landscape of proteins. While AlphaFold and ESMFold were originally developed and trained as \emph{regression} models that predict a single best protein structure for a given MSA or sequence input, we develop a strategy for repurposing them as (sequence-conditioned) \emph{generative} models of protein structure. This synthesis relies on the key insight that iterative denoising frameworks (such as diffusion and flow-matching) provide a general recipe for converting regression models to generative models with relatively little modification to the architecture and training objective. Unlike inference-time MSA ablation, this strategy applies equally well to PLM-based predictors and can be used to train or fine-tune on arbitrary ensembles.
\begin{figure*}
    \centering
    \includegraphics[width=\textwidth]{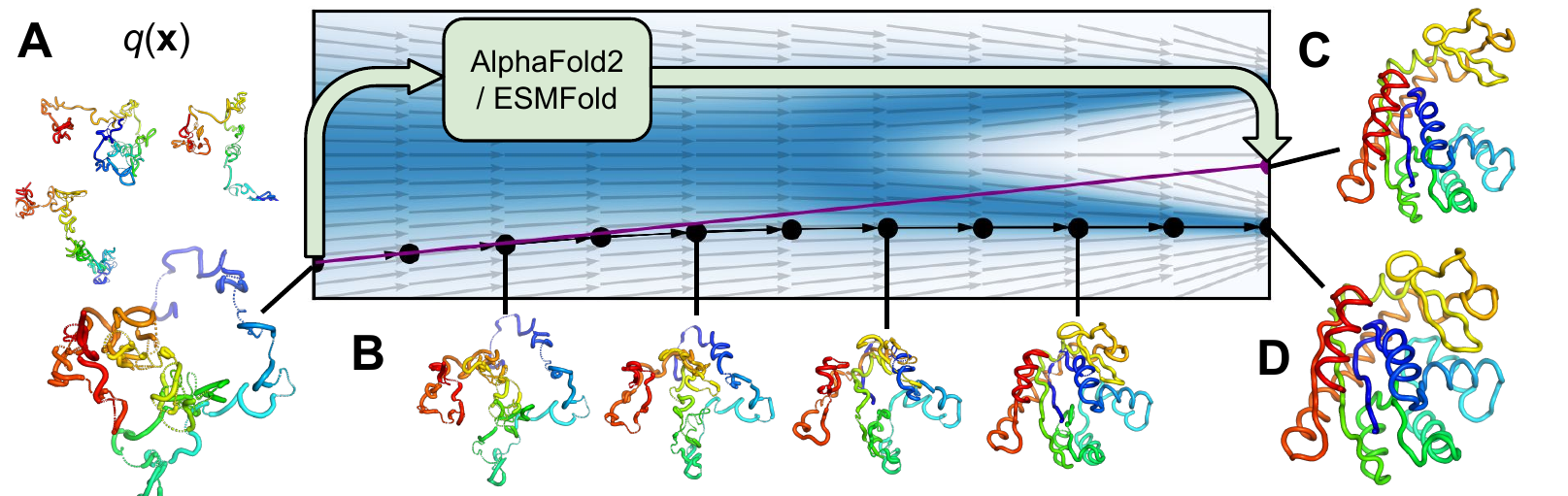}
    \vspace{-1.75\baselineskip}
    \caption{\textbf{Conceptual overview of Alpha\textsc{Flow} / ESM\textsc{Flow}.} (A) Samples are drawn from a harmonic (polymer-like) prior. (B) The sample is progressively refined or denoised under a flow field controlled by the structure prediction model (AlphaFold or ESMFold). (C) At each step, the denoised structure prediction parameterizes the direction of the flow and we interpolate the current sample towards it. (D) The final prediction is a sample from the learned distribution of structures.}
    \vspace{-0.75\baselineskip}
    \label{fig:overview}
\end{figure*}

While flow matching has been well established for images, its application to protein structures remains nascent \citep{bose2023se}. Hence, we develop a custom flow matching framework tailored to the architecture and training practices of AlphaFold and ESMFold. Our framework leverages the polymer-structured prior distribution from {harmonic diffusion} \citep{jing2023eigenfold}, but improves over it by defining a scale-invariant noising process resilient to missing and cropped residues. These improvements directly result from the increased modeling flexibility offered by flow matching and contribute to the performance of our method.

We demonstrate the performance of our flow-matching variants of AlphaFold and ESMFold---named Alpha\textsc{Flow} and ESM\textsc{Flow}---in two distinct settings. First, after fine-tuning these models only on structures from the PDB, we substantially surpass the precision-diversity Pareto frontier of MSA ablation baselines on a test set of recently deposited conformationally heterogeneous proteins. Second, we showcase the ability to learn from ensembles beyond the PDB by further training on the ATLAS dataset \citep{vander2023atlas} of molecular dynamics simulations. When evaluated on test proteins structurally dissimilar from the training set, Alpha\textsc{Flow} substantially surpasses the MSA baselines in the prediction of conformational flexibility, distributional modeling of atomic positions, and replication of higher-order ensemble observables such as intermittent contacts and solvent exposure. Furthermore, when a static PDB structure is provided as a template, sampling from Alpha\textsc{Flow} provides faster wall-clock convergence to many equilibrium properties than running molecular dynamics (MD) simulation starting from that structure. Thus, our method can be used in place of expensive simulations to diversify and obtain equilibrium ensembles of solved protein structures.

\section{Background}

\textbf{Protein structure prediction.} The modern approach for protein structure prediction was pioneered by AlphaFold \citep{jumper2021highly}, which takes as input (1) the protein sequence, (2) a MSA of evolutionarily related sequences, and optionally (3) a template structure of a related protein, and predicts the all-atom 3D coordinates of single protein structure. AlphaFold was developed and trained in an end-to-end fashion under a regression-like FAPE loss with structures from the PDB. Later works, such as ESMFold \citep{lin2023evolutionary} and OmegaFold \citep{wu2022high}, modified the pipeline by substituting the MSA with embeddings from a protein language model (PLM) and eschewing the template input, but otherwise kept the same architecture and training framework as AlphaFold.

\textbf{Modeling protein ensembles.} In the post-AlphaFold era, several works have emphasized diversifying highly accurate single-structure predictions to reflect underlying conformational heterogeneity \citep{lane2023protein,chakravarty2022alphafold2,saldano2022impact, xie2023can, brotzakis2023alphafold, bryant2023structure, porter2023colabfold}. Most prominently, \citet{del2022sampling} demonstrated that multiple functional states could be obtained by subsampling the MSA input to AlphaFold. Since then, MSA subsampling has become the \emph{de-facto} standard methodology and has been employed to study conformational states of kinases \citep{faezov2023alphafold2, herrington2023exploring, casadevall2023alphafold2}, variant effects on conformational states \citep{da2023predicting}, and to seed molecular dynamics simulations \citep{vani2023alphafold2}. Alternative approaches have also been proposed in the form of point mutations to the MSA \citep{stein2022speach_af, stein2023rosetta} and MSA clustering \citep{wayment2023predicting}. .

An emerging line of work seeks to directly train sequence-to-structure generative models of protein ensembles. EigenFold \citep{jing2023eigenfold} and Distributional Graphormer \citep{zheng2023towards} use harmonic diffusion and $SE(3)$ diffusion \citep{yim2023se}, respectively, to generate ensembles. SENS \citep{lu2023score} is a local generative model that diversifies single starting structures via local exploration of the conformational landscape. However, these models have yet to show convincing validations or comparisons with MSA subsampling methods on PDB test sets.

A related but separate line of work has focused on learning generative models of Boltzmann distributions as proxies for expensive molecular dynamics simulation. These models were initially conceived as normalizing flows that provided exact  likelihoods and thus a means to train with energies and reweigh samples at inference time \citep{noe2019boltzmann, kohler2021smooth, midgley2022flow, abdin2023pepflow, felardos2023designing}. However, these normalizing flows have proven difficult to scale beyond small molecules and toy systems. More recently, the proliferation of diffusion models has shifted the focus of this line of work towards scalability and generalization \citep{arts2023two, zheng2023towards} rather than exact likelihoods. Our method, when trained on MD ensembles, can be viewed as belonging to this new generation of Boltzmann-targeting generative models.

\textbf{Flow matching} \citep{lipman2022flow, albergo2022building, albergo2023stochastic, liu2022flow} is a generative modeling paradigm that resembles and builds upon the significant success of diffusion models \citep{ho2020denoising, song2021score} in image and molecule domains. The fundamental object in flow matching is a conditional probability path $p_t(\bfx \mid \bfx_1), t \in [0, 1]$: a family of densities conditioned on a data point $\bfx_1 \sim p_\text{data}$ which interpolates between a shared prior distribution $p_0(\bfx \mid \bfx_1) = q(\bfx)$ and an approximate Dirac $p_1(\bfx \mid \bfx_1) \approx \delta(\bfx - \bfx_1)$. Given a conditional vector field $u_t(\bfx \mid \bfx_1)$ that generates the time evolution of $p_t(\bfx \mid \bfx_1)$, one then learns the \emph{marginal vector field} with a neural network:
\begin{equation}\label{eq:flow_matching}
    {\hat v}(\bfx, t; \theta) \approx v(\bfx, t) := \mathbb{E}_{\bfx_1 \sim p_t(\bfx_1 \mid \bfx)}[u_t(\bfx \mid \bfx_1)]
\end{equation}
At convergence, the learned vector field $\hat v(\bfx,t;\theta)$ is a neural ODE that evolves the prior distribution $q(\bfx)$ to the data distribution $p_\text{data}(\bfx)$. Score-matching in diffusion models can be seen as a special case of flow matching; however, as discussed in Section~\ref{sec:diffusion}, flow matching circumvents certain difficulties that would otherwise arise with diffusion.

\section{Method}
\subsection{AlphaFold as a Denoising Model}

\begin{figure}
    \centering
    \includegraphics[width=\columnwidth]{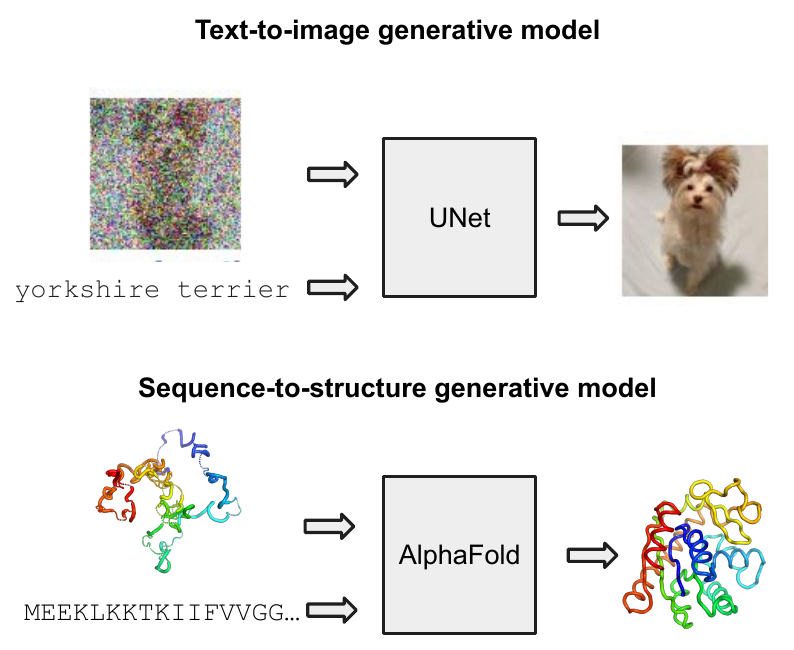}
    \vspace{-1.5\baselineskip}
    \caption{\textbf{AlphaFold as a denoising model.} Just as (diffusion-based) text-to-image generative models are simply neural networks that denoise images (with text input), a modified AlphaFold that ingests noisy structures and predicts clean structures (with sequence input) immediately provides a sequence-to-structure generative model---when trained under an appropriate framework.}
    \vspace{-1\baselineskip}
    \label{fig:af-as-denoising}
\end{figure}

Given a protein sequence $A$ of amino acid tokens, our objective is to model the distribution $p(\bfx \mid A)$ over 3D coordinates $\bfx \in \mathbb{R}^{3\times N}$ which represents the structural ensemble of that protein sequence. Considering the enormous intellectual efforts that went into a \emph{deterministic} sequence-to-structure model (i.e., AlphaFold), developing a distributional model of equivalent accuracy and generalization ability would appear to pose a considerable challenge. Our solution is to leverage recent conceptual advances in generative modeling in order to simply \emph{repurpose} AlphaFold---nearly out of the box---as a generative model. 

Consider, for example, the (simplified) architecture of prototypical text-to-image diffusion models \citep{ho2020denoising, rombach2022high}, which aim to model conditional distributions $p(\bfx \mid s)$ of images $\bfx$ conditioned on text prompt $s$. At the heart of these models lies a \emph{denoising neural network} (e.g., a UNet) which ingests a noisy image, along with a text prompt, to predict a clean image. Conditioned on these inputs, such models are otherwise are trained with simple, regression-like MSE objectives. Analogously, a protein structure predictor trained on a regression-like loss---like AlphaFold or ESMFold---can be converted to a denoising model simply by supplying an additional, \emph{noisy} structure input (Figure~\ref{fig:af-as-denoising}). Not coincidentally, this is reminiscent of the idea of \emph{template structures} employed by certain AlphaFold workflows. Thus, we develop an input embedding module very similar to AlphaFold's template embedding stack and prepend it to the pairwise folding trunks of AlphaFold and ESMFold (details in Appendix~\ref{app:input_embedder}). By doing so, we obtain structure denoising architectures that are thin wrappers around well-validated single-structure predictors.

With these architectural modifications, we are ready to plug AlphaFold and ESMFold into any iterative denoising-based generative modeling framework. Next, we will see how this concretely applies to flow matching for protein ensembles.

\subsection{Flow Matching for Protein Ensembles} \label{sec:flow_matching}
Designing a flow-matching generative framework amounts \clearpage to the choice of a conditional probability path $p_t(\bfx \mid \bfx_1)$ and its corresponding vector field $u_t(\bfx \mid \bfx_1)$. Inspired by the interpolant-based perspective on flow matching \citep{albergo2022building}, we define the conditional probability path by sampling noise $\bfx_0$ from the prior $q(\bfx_0)$ and interpolating linearly with the data point $\bfx_1$:
\begin{equation} \label{eq:fm_start}
    \bfx \mid \bfx_1, t = (1 - t) \cdot \bfx_0 + t \cdot \bfx_1, \quad \bfx_0 \sim q(\bfx_0)
\end{equation}
This probability path is associated with the vector field
\begin{equation}\label{eq:fm_vector_field}
    u_t(\bfx \mid \bfx_1) = (\bfx_1 - \bfx)/(1-t)
\end{equation}
which matches the CondOT path and field proposed in (for example) \citet{pooladian2023multisample}. Customarily, we then learn a neural network to approximate the marginal vector field according to Equation~\ref{eq:flow_matching}. However, if we instead define a neural network ${\hat{\bfx}_1}(\bfx, t;\theta)$ and reparameterize via 
\begin{equation} \label{eq:reparameterization}
    {\hat v}(\bfx, t; \theta) = ({\hat{\bfx}_1}(\bfx, t;\theta) - \bfx)/(1-t)
\end{equation}
then rearrangements of Equations~\ref{eq:flow_matching} and \ref{eq:reparameterization} reveal that we can equivalently learn the expectation of $\bfx_1$: 
\begin{equation}\label{eq:x1_prediction}
    {\hat \bfx}_1(\bfx, t;\theta) \approx \mathbb{E}_{\bfx_1 \sim p_t(\bfx_1 \mid \bfx)}[\bfx_1]
\end{equation}
This reparameterization is identical---up to the choice of probability path $p_t(\bfx_1 \mid \bfx)$---to that employed for image diffusion models \citep{ho2020denoising}. In our setting, since $\bfx_1$ refers to samples from the data distribution (i.e., protein structures), this allows the AlphaFold-based architectures discussed previously to be immediately used as the the denoising model ${\hat{\bfx}_1}(\bfx, t;\theta)$, with $\bfx$ as the noisy input and $t$ as an additional time embedding.

To apply flow matching to protein structures, we describe a structure by the 3D coordinates of its $\beta$-carbons ($\alpha$-carbon for glycine): $\bfx \in \mathbb{R}^{N \times 3}$. (We choose $\beta$-carbons because these are the inputs to the template embedding stack.) We then define the prior distribution $q(\bfx)$ over the positions of these $\beta$-carbons to be a \emph{harmonic prior} \citep{jing2023eigenfold}:
\begin{equation} \label{eq:fm_end}
    q(\bfx) \propto \exp\left[- \frac{\alpha}{2}\sum_{i=1}^{N-1} \lVert\bfx_i - \bfx_{i+1} \rVert^2 \right]
\end{equation}
This prior ensures that samples along the conditional probability path, and hence inputs to the neural network, always remain polymer-like, physically plausible 3D structures.

The parameterization of learning the conditional expectation of $\bfx_1$ (Equation~\ref{eq:x1_prediction}) suggests that the neural network should be trained with an MSE loss. However, there are several issues with this direct approach. (1) The structure prediction networks not only predict $\beta$-carbon coordinates, but also all-atom coordinates and residue frames. (2) The input to the network is $SE(3)$-invariant by design, which makes training with MSE loss unsuitable without further correction (Appendix~\ref{app:flow_matching}). Finally, (3) the networks obtain best performance (and were orginally trained) with the $SE(3)$-invariant Frame Aligned Point Error (FAPE) loss.  To reconcile these issues with the flow-matching framework, we redefine the space of protein structures to be the \emph{quotient} space $\mathbb{R}^{3 \times N}/SE(3)$, with the prior distribution projected to this space. We redefine the interpolation between two points in this space to be linear interpolation in $\mathbb{R}^3$ after RMSD-alignment. Further, because the quotient space is no longer a vector space, there is no longer a notion of ``expectation" of a distribution; instead, we aim to learn the more general Fr\'echet mean of the conditional distribution $p(\bfx_1 \mid \bfx)$:
\begin{equation}
    {\hat \bfx}_1(\bfx, t;\theta) \approx  \min_{{\hat\bfx}_1}  \mathbb{E}_{\bfx_1 \sim p_t(\bfx_1 \mid \bfx)}\left[\FAPE^2(\bfx_1, \hat\bfx_1) \right]
\end{equation}
where we leverage the property that FAPE is a valid metric \citep{jumper2021highly} to define a Fr\'echet mean. To learn this target, we use a training loss identical to the original FAPE, except now \emph{squared}. The final result for the training and inference procedures are provided in Algorithms~\ref{alg:training} and \ref{alg:inference}. An important implication of this modified framework is that while our model is faithfully supervised on all-atom coordinates, it technically is learning the \emph{distribution} only over $\beta$-carbon coordinates. These procedures and their subtleties are more fully discussed in Appendix~\ref{app:flow_matching}.

\begin{algorithm}[H]
\caption{\textsc{Training}}\label{alg:training}
\begin{algorithmic}
\STATE \textbf{Input:} Training examples of structures, sequences, and MSAs $\{(S_i, A_i, M_i)\}$
\FORALL{$(S_i, A_i, M_i)$}
\STATE Extract $\bfx_1 \gets \BetaCarbons(S_i)$ 
\STATE Sample $\bfx_0 \sim \HarmonicPrior(\length(A_i))$ 
\STATE Align $\bfx_0 \gets \RMSDAlign(\bfx_0, \bfx_1)$ 
\STATE Sample $t \sim \Uniform[0, 1]$ 
\STATE     Interpolate $\bfx_t \gets t \cdot \bfx_1 + (1-t) \cdot \bfx_0$ 
\STATE Predict $\hat S_i \gets \AlphaFold(A_i, M_i, \bfx_t, t)$ 
\STATE  Optimize loss $\mathcal{L} = \FAPE^2(\hat S_i, S_i)$ 
\ENDFOR
\end{algorithmic}
\end{algorithm}

\begin{algorithm}[H]
\caption{\textsc{Inference}}\label{alg:inference}
\begin{algorithmic}
\STATE \textbf{Input:} Sequence and MSA $(A, M)$
\STATE \textbf{Output:} Sampled all-atom structure $\hat S$
\STATE Sample $\bfx_0 \sim \HarmonicPrior(\length(A))$ 
\FOR{$n \gets 0$ to $N - 1$}
\STATE    Let $t \gets n / N$ and $s \gets t + 1/N$ \;
\STATE    Predict $\hat S \gets \AlphaFold(A, M, \bfx_t, t)$ \;
\IF{$n = N-1$}
\STATE \textbf{return} $\hat S$ 
\ENDIF
\STATE    Extract $\hat \bfx_1 \gets \BetaCarbons(\hat S)$ \;
\STATE    Align $\bfx_t \gets \RMSDAlign(\bfx_t, \hat \bfx_1)$ \;
\STATE    Interpolate $\bfx_s \gets \frac{s - t}{1-t} \cdot \hat \bfx_1 + \frac{1-s}{1-t} \cdot \bfx_t$ \;
\ENDFOR
\end{algorithmic}
\end{algorithm}

\subsection{Comparison with Diffusion}\label{sec:diffusion}
Since our flow matching framework involves defining and reversing a noising process, it bears a number of similarities with harmonic diffusion for protein structures \citep{jing2023eigenfold}, which converges to the same prior distribution. However, as a more general framework, flow matching offers two key advantages. First, harmonic diffusion converges to the prior distribution only in the infinite-time limit, and at a rate that \emph{depends on the data dimensionality}, i.e., protein size. This causes inference-time distributional shifts when training only on crops of relatively small size, as is the case with AlphaFold and ESMFold. On the other hand, in flow matching, the prior distribution is imposed as a boundary condition at time $t=0$ for all dimensionalities. Second, flow matching provides an easy means to deal with missing (gap) residues---which are very common in the PDB---by simply omitting them in the interpolation. In contrast, harmonic diffusion induces dependencies across atomic positions and hence requires data imputation for missing residues. We discuss these aspects (with additional theoretical results) further in Appendix~\ref{app:harmonic_diff}. 

\section{Experiments}

\begin{figure*}
    \centering
    \includegraphics[width=0.35\textwidth]{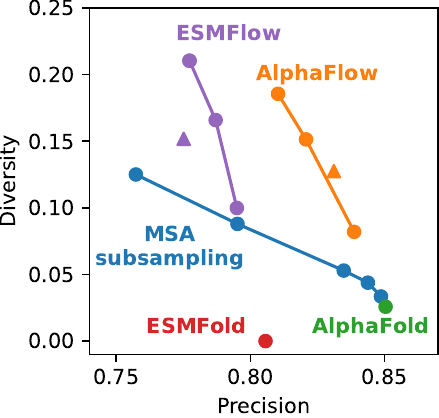} \hspace{12pt}
    \includegraphics[width=0.35\textwidth]{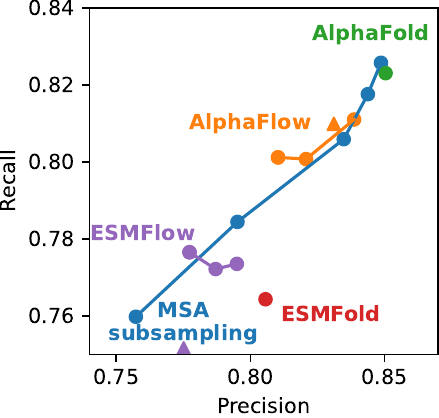} 
    \vspace{-\baselineskip}
    \caption{\textbf{Evaluation on PDB ensembles---}precision-diversity (\emph{left}) and precision-recall (\emph{right}) curves for all benchmarked methods (median taken over 100 test targets). The MSA subsampling curve is traced by reducing MSA depth (max 512, min 48) and joins to AlphaFold as they share the same weights (AlphaFold by default subsamples MSAs to a maximum depth of 1024 and thus has nonzero diveristy, unlike ESMFold). The Alpha\textsc{Flow} / ESM\textsc{Flow} curves are traced by truncating the initial steps of flow matching (described in Appendix~\ref{app:training_inference}). Distilled models are marked by $\blacktriangle$. Tabular data is shown in Appendix~\ref{app:pdb_results}, Table~\ref{tab:pdb_results}}
    \label{fig:pdb-results-main}
    \vspace{-\baselineskip}
\end{figure*}

\subsection{Training Regimen}

We fine-tune all weights of AlphaFold and ESMFold on the PDB with our flow matching framework, starting from their publicly available pretrained weights. We use OpenFold \citep{ahdritz2022openfold} for the architecture implementation and training pipeline and OpenProteinSet \citep{ahdritz2023openproteinset} for training MSAs. Adhering to the original works, we use a training cutoff of May 1, 2018 and May 1, 2020 for AlphaFold and ESMFold, respectively. At the conclusion of this stage of training (1.28M and 720k examples, respectively), we obtain flow-matching variants of AlphaFold and ESMFold which we call \textbf{Alpha\textsc{Flow}} and \textbf{ESM\textsc{Flow}}. 

Next, to demonstrate and assess the ability of our method to learn from MD ensembles, we continue fine-tuning both models on the ATLAS dataset of all-atom MD simulations \citep{vander2023atlas}, which consists of 1390 proteins chosen for structural diversity by ECOD domain classification \citep{schaeffer2017ecod}. Using training and validation cutoffs of May 1, 2018 and May 1, 2019, we obtain train/val/test splits of 1265/39/82 ensembles (2 excluded due to length). After 43k and 27k additional training examples, respectively, we obtain MD-specialized variants of our model which we call \textbf{Alpha\textsc{Flow}-MD} and \textbf{ESM\textsc{Flow}-MD}. We also train variants of these models (\textbf{+Templates}) which accept the PDB structure that initialized the simulation as input using a copy of the input embedding module.

Because flow matching is an iterative generative process, sampling a single structure requires many forward passes of AlphaFold and ESMFold (up to 10 in our experiments), which can be somewhat expensive. To accelerate this process, we explore variants of all models where the generative process is \emph{distilled} into a single forward pass (details in Appendix~\ref{app:training_inference}). While distillation has been explored for diffusion models \citep{salimans2022progressive, song2023consistency, yin2023one}, this is (to our knowledge) the first demonstration of distillation in a protein or flow-matching setting.

\subsection{PDB Ensembles}

We first examine the ability of Alpha\textsc{Flow} and ESM\textsc{Flow} to sample diverse conformations  of proteins deposited in the Protein Data Bank (PDB). To do, we construct a test set of 100 proteins deposited after the AlphaFold training cutoff (May 1, 2018) with multiple chains and evidence of conformational heterogeneity (details in Appendix~\ref{app:datasets}). For each protein, we sample 50 predictions with (1) unmodified AlphaFold/ESMFold (2) AlphaFold with varying degrees of MSA subsampling and (3) Alpha\textsc{Flow}/ESM\textsc{Flow}, with varying degrees of flow truncation in order to tune the amount of diversity (Appendix~\ref{app:training_inference}). Each set of predictions is evaluated on three metrics: \textbf{precision}---the average lDDT$_{\text{C}\alpha}$ from each prediction to the closest crystal structure; \textbf{recall}---the average lDDT$_{\text{C}\alpha}$ from each crystal structure to the closest prediction; and \textbf{diversity}---the average dissimilarity (1-lDDT$_{\text{C}\alpha}$) between pairs of predicted structures.

The median results across the 100 test targets are shown in Figure~\ref{fig:pdb-results-main}. Alpha\textsc{Flow}, similar to MSA subsampling, increases the prediction diversity relative to the unmodified AlphaFold at the cost of reduced precision. However, the variants of Alpha\textsc{Flow} trace a substantially superior Pareto frontier relative to MSA subsampling. In some cases, PCA of the ground truth and predicted ensembles (Appendix~\ref{app:pdb_results}) offers an explanation for this result: in MSA subsampling, the ensembles drift away from the true structures as the input signal is ablated, whereas the Alpha\textsc{Flow} predictions remain clustered around the ground truth conformations while reaching the same or greater levels of diversity. In terms of precision and recall, Alpha\textsc{Flow} exhibits very similar behavior to MSA subsampling. Somewhat surprisingly, neither method is able to meaningfully improve aggregate recall relative to baseline AlphaFold, showing that they generally do not succeed in increasing the coverage of experimentally determined PDB structures, or (more optimistically) that the predicted conformational changes have yet to be experimentally observed. Selected cases of conformational changes successfully modeled by Alpha\textsc{Flow} are visualized in Appendix~\ref{app:pdb_results}; Figure~\ref{fig:pdb_conf_changes}.

Overall (as expected), ESMFold and ESM\textsc{Flow} exhibit reduced precision relative to AlphaFold-family methods. However, ESM\textsc{Flow} is able to inject substantial diversity relative to baseline ESMFold---which, unlike AlphaFold, is completely deterministic---and increase the recall at little to no cost in precision. Note that this test set includes some proteins deposited before the ESMFold cutoff; results on a later sub-split are similar (Appendix~\ref{app:pdb_results}; Table~\ref{tab:pdb_results}).

\subsection{Molecular Dynamics Ensembles}

\begin{table*}[]
    \centering
    \caption{\textbf{Evaluation on MD ensembles}. For each method, we compare the predicted ensemble with the ground truth MD ensemble according to various metrics, detailed in the main text. For protein flexibility and RMSF, the ground truth values (from the MD ensembles) are in parenthesis. When applicable, the median across the 82 test ensembles is reported. See Appendix~\ref{app:md_results} for ESM\textsc{Flow} results. $r$: Pearson correlation; $\rho$: Spearman correlation; $J$: Jaccard similarity; $\mathcal{W}_2$: 2-Wasserstein distance.}
    \label{tab:md_results}
    \begin{small}
    \begin{tabular}{clccccccc|cc}
    \toprule
    & &\multicolumn{2}{c}{Alpha\textsc{Flow}-MD} & \multicolumn{4}{c}{MSA subsampling} & & \multicolumn{2}{c}{AFMD+Templates} \\ 
    \cmidrule(lr){3-4} \cmidrule(lr){5-8}  \cmidrule(lr){10-11}
    & & Full & Distilled & 32 & 48 & 64 & 256 & AlphaFold & Full & Distilled \\
    \midrule
    \multirow{5}{*}{\makecell{Predicting\\flexibility}}
    & Pairwise RMSD ($=$ 2.90) &       \textbf{2.89} &     1.94 &    4.40 &    2.34 &    1.67 &     0.72 &       0.58 &       2.18 & 1.73 \\
    & Pairwise RMSD $r$ $\uparrow$     &       \textbf{0.48} &     \textbf{0.48} &    0.03 &    0.12 &    0.22 &     0.15 &       0.10 &       0.94 & 0.92 \\
    & All-atom RMSF ($=$1.70) &       \textbf{1.68} &     1.28 &    5.38 &    2.29 &    1.17 &     0.49 &       0.31 &       1.31 & 1.00 \\
    & Global RMSF $r$ &       \textbf{0.60} &     0.54 &    0.13 &    0.23 &    0.29 &     0.26 &       0.21 &       0.91 & 0.89 \\
    & Per-target RMSF $r$      &       \textbf{0.85} &     0.81 &    0.51 &    0.52 &    0.51 &     0.55 &       0.52 &       0.90 & 0.88 \\
    \midrule
    \multirow{6}{*}{\makecell{Distributional\\accuracy}} 
    & Root mean $\mathcal{W}_2$-dist. $\downarrow$ &       \textbf{2.61} &     3.70 &    6.15 &    5.32 &    4.28 &     3.62 &       3.58 &       1.95 & 2.18 \\
    & $\hookrightarrow$ Translation contrib. $\downarrow$ &       \textbf{2.28} &     3.10 &    5.22 &    3.92 &    3.33 &     2.87 &       2.86 &       1.64 & 1.74 \\
    & $\hookrightarrow$ Variance contrib. $\downarrow$ &       \textbf{1.30} &     1.52 &    3.55 &    2.49 &    2.24 &     2.24 &       2.27 &       1.01 & 1.25 \\
    & MD PCA $\mathcal{W}_2$-dist. $\downarrow$      &       \textbf{1.52} &     1.73 &    2.44 &    2.30 &    2.23 &     1.88 &       1.99 &       1.25 & 1.41 \\
    & Joint PCA $\mathcal{W}_2$-dist. $\downarrow$ &       \textbf{2.25} &     3.05 &    5.51 &    4.51 &    3.57 &     3.02 &       2.86 &       1.58 & 1.68 \\
    & \% PC-sim $>0.5$ $\uparrow$      &       \textbf{44} &     34 &    15 &    18 &    21 &     21 &       23 &       44 & 43 \\
    \midrule
    \multirow{4}{*}{\makecell{Ensemble\\observables}}  
    & Weak contacts $J$ $\uparrow$      &       \textbf{0.62} &     0.52 &    0.40 &    0.40 &    0.37 &      0.30 &       0.27 &       0.62 & 0.51 \\
    & Transient contacts $J$ $\uparrow$ &       \textbf{0.41} &     0.28 &    0.23 &    0.26 &    0.27 &      0.27 &       0.28 &       0.47 & 0.42 \\
    & Exposed residue $J$ $\uparrow$    &       \textbf{0.50} &     0.48 &    0.34 &    0.37 &    0.37 &     0.33 &       0.32 &       0.50 & 0.47 \\
    & Exposed MI matrix $\rho$ $\uparrow$ &       \textbf{0.25} &     0.14 &    0.14 &    0.11 &    0.10 &     0.06 &       0.02 &       0.25 & 0.18 \\
    \bottomrule
    \end{tabular}
    \vspace{-\baselineskip}
    \end{small}
\end{table*}

We next evaluate the ability of Alpha\textsc{Flow} and ESM\textsc{Flow} to generate proxy MD ensembles for the 82 test proteins in the ATLAS database. These test proteins have minimal structural overlap with the training ensembles, providing a stringent test of generalization. For each target, we sample 250 predictions with each method and probe their similarity to the MD ensembles via a series of assessments, grouped under three broad categories of increasing difficulty: (1) predicting flexibility, (2) distributional accuracy, and (3) ensemble observables. Unless otherwise noted, we focus on Alpha\textsc{Flow} ensembles generated with MSA input alone (i.e., no PDB templates). Main results are presented in Table~\ref{tab:md_results} and Figure~\ref{fig:md_results}; further results (e.g. ESM\textsc{Flow}, comparisons with normal mode analysis, and ablations) and ensemble visualizations can be found in Appendix~\ref{app:md_results}. We note that our evaluations are inherently limited to phenomena accessible within the ATLAS simulation timescales; we do not assess if our model captures slower conformational changes, which remain a key area for future work. 

\textbf{Q1: Is ensemble flexibility predictive of true protein flexibility?} For each ensemble, we quantify the protein flexibility as the average C$\alpha$-RMSD between any pair of conformations. By this metric, the Alpha\textsc{Flow} ensembles have the strongest Pearson correlation with the ground truth and matches the aggregate level of diversity in the MD ensembles. In contrast, MSA subsampling is unable to reach the same level of diversity while retaining any predictive power. Similar results hold when considering atomic-level flexibility in terms of root mean square fluctuation (RMSF), both when pooled globally and pooled per-target. Remarkably, Alpha\textsc{Flow} attains a median Pearson correlation of 0.85 between modeled and predicted RMSFs within a target, while no level of MSA subsampling is able to meaningfully exceed baseline AlphaFold on this metric.

\textbf{Q2: Are the atomic positions distributionally accurate?} To generalize the all-atom RMSD metric to ensembles, define the \emph{root mean Wasserstein distance} (RMWD) between ensembles $\mathcal{X}, \mathcal{Y}$ as
\begin{equation}\label{eq:rmwd}
    \RMWD(\mathcal{X}, \mathcal{Y}) = \sqrt{\frac{1}{N} \sum_{i=1}^N \mathcal{W}_2^2\left(\mathcal{N}[\mathcal{X}_i], \mathcal{N}[\mathcal{Y}_i]\right)}
\end{equation} 
where $\mathcal{N}[\mathcal{X}_i]$ are 3D-Gaussians fit to the positional distribution of the $i$th atom in ensemble $\mathcal{X}$ (this reduces to RMSD with a single structure). By this metric, Alpha\textsc{Flow} ensembles are more accurate than any level of MSA subsampling. Decomposition of the RMWD into a translation contribution and variance contribution (Appendix~\ref{app:evaluation}) reveals that Alpha\textsc{Flow} slightly improves on AlphaFold in predicting the mean position of atoms, and substantially outperforms MSA subsampling in modeling the variance.

The joint distribution of C$\alpha$ positions reveals collective motions and provides a more stringent test of distributional accuracy. We project this joint distribution onto the first two principal components from PCA---computed from the MD ensemble alone or from equally weighting the MD and predicted ensembles---and compute the $\mathcal{W}_2$-distance (in units of \AA\ RMSD) between the predicted and true ensembles in this space. We also compute the (unsigned) \emph{cosine similarity} between the top principal components of the predicted and true ensembles and consider the dominant motion to be successfully modeled if this similarity $>0.5$. By all of these metrics, Alpha\textsc{Flow} markedly improves over MSA subsampling, and in particular nearly doubles the success rate for obtaining $>0.5$ cosine similarity.

\textbf{Q3: Are complex ensemble observables faithfully reproduced?} MD ensembles are often intended for downstream analysis of observables such as intermittent contacts and solvent exposure, often associated with thermal fluctuations around the low-energy crystal structure \citep{vogele2022systematic}. To probe if we model these properties accurately, for each ensemble we identify the set of \emph{weak contacts} and \emph{transient contacts}, defined as those C$\alpha$ pairs which are in contact (respectively, not in contact) in the crystal structure but which dissociate (respectively, associate) in $> 10\%$ of ensemble structures, with a 8 \AA\ threshold. We then compute the the Jaccard similarity of the sets produced by each method with the ground truth sets. We repeat the same analysis with the set of \emph{cryptically exposed residues}---those whose sidechains are buried in the crystal structure but exposed to solvent in $> 10\%$ of ensemble structures---which are a key feature in the identification of cryptic pockets in drug discovery \citep{meller2023predicting}. Going further, for each pair of residues we compute the \emph{mutual information} (MI) between their (binary) exposure states, yielding a MI matrix for each ensemble. Such matrices are an important in the so-called \emph{exposon analysis} of protein dynamics, e.g., for collective motions and allostery \citep{porter2019cooperative}. We then compute the Spearman correlation between the values of MI matrices from the MD and generated ensembles. Impressively, for all of these analyses, Alpha\textsc{Flow} substantially outperforms MSA subsampling; we emphasize that these are complex properties to emulate involving sidechains and different parts of the protein (Figure~\ref{fig:md_results} and Appendix~\ref{app:md_results}).

\begin{figure}[h]
    \centering
    \includegraphics[width=\columnwidth]{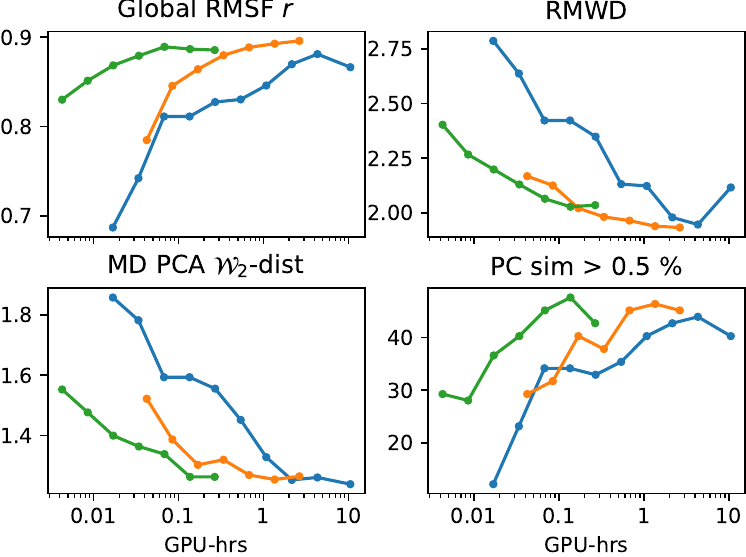}
    \vspace{-1.5\baselineskip}
    \caption{\textbf{Efficiency of Alpha\textsc{Flow} vs replicate MD simulations.} Alpha\textsc{Flow}+Templates with varying number of samples with distillation (\textcolor[HTML]{2ca02c}{green}) and without distillation (\textcolor[HTML]{ff7f0e}{orange}); MD with varying trajectory lengths in \textcolor[HTML]{1f77b4}{blue}. See Appendix~\ref{app:evaluation} for further experimental details and Appendix~\ref{app:md_results} for further results.}
    \label{fig:md-convergence}
    \vspace{-\baselineskip}
\end{figure}

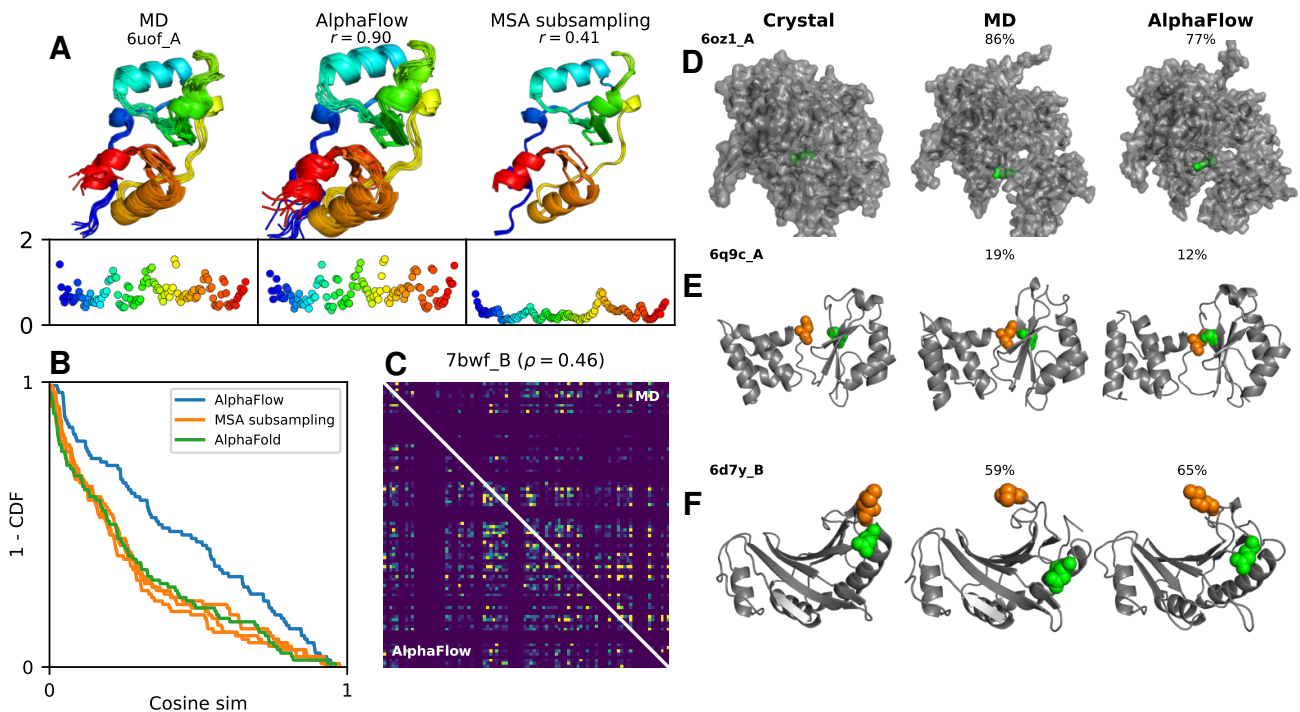
\begin{figure*}[t]
    \centering
    \begin{tikzpicture}
    \node[anchor=south west,inner sep=0] (image) at (0,0) 
    {\includegraphics[width=\textwidth]{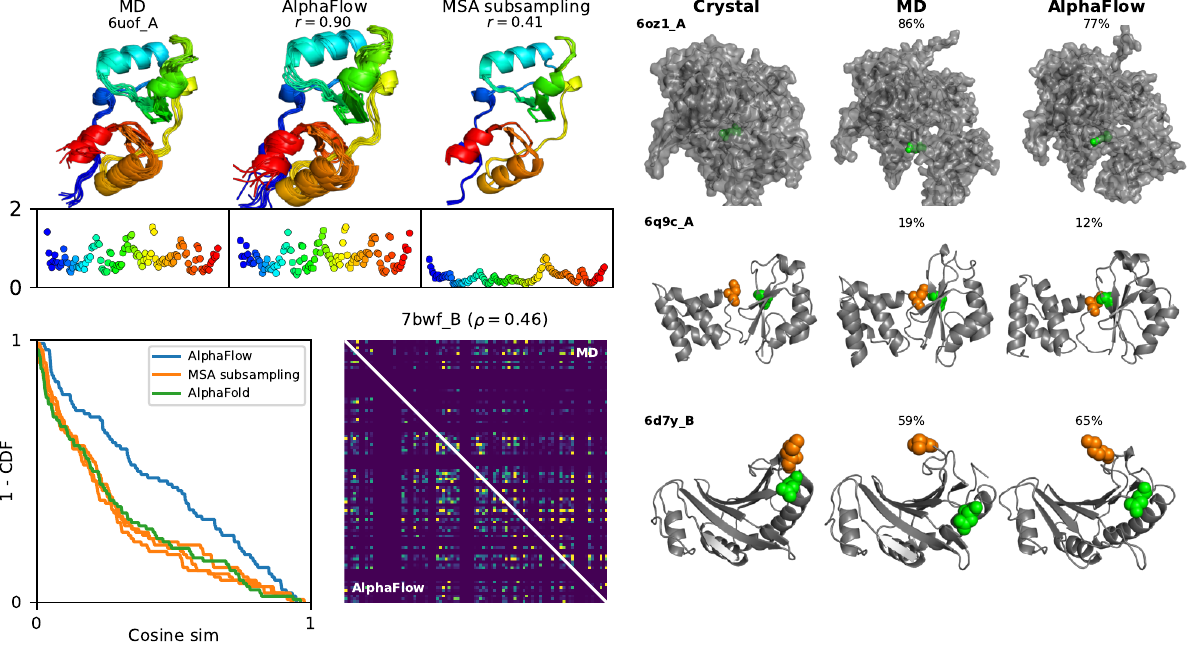}};
    \begin{scope}[x={(image.south east)},y={(image.north west)}]
    \node at (0.04,0.95) {\large  \bf \sffamily A};
    \node at (0.04,0.5) {\large \bf \sffamily B};
    \node at (0.3,0.5) {\large \bf \sffamily C};
    \node at (0.53,0.93) {\large \bf \sffamily D};
    \node at (0.53,0.61) {\large \bf \sffamily E};
    \node at (0.53,0.30) {\large \bf \sffamily F};
    \end{scope}
    \end{tikzpicture}
    
    \caption{\textbf{MD evaluations visualized.} (\textbf{A}) Ensembles of PDB ID \texttt{6uof\_A} (transcriptional regulator from \emph{Streptococcus pneumoniae}) from MD, Alpha\textsc{Flow}, and MSA subsampling (depth 48), with C$\alpha$ RMSF by residue index shown in insets. (\textbf{B}) $1-CDF$ of the distribution of (unsigned) cosine similarities between the top principal components of the predicted ensemble versus the MD ensemble. (\textbf{C}) Solvent exposure mutual information matrices computed from the ground truth MD ensemble and Alpha\textsc{Flow} ensemble for target PDB ID \texttt{7bwf\_B} (antitoxin from \emph{Staphylococcus aureus}). (\textbf{D, E, F}) Deviations from the crystal structure in the MD simulation, corresponding to ensemble observables, which are correctly sampled by Alpha\textsc{Flow}. The probability of occurence in each ensembles is shown. (\textbf{D}) solvent exposure of a buried residue in PDB ID \texttt{6oz1\_A} (carboxylate reductase from \emph{M. chelonae}). (\textbf{E}) association of a transient residue contact in PDB ID \texttt{6q9c\_A} (NADH-quinone oxidoreductase subunit E from \emph{Aquifex aeolicus}). (\textbf{F}) dissociation of a weak residue contact in PDB ID \texttt{6d7y\_B} (immune protein from \emph{Enterobacter cloacae}). Additional examples in Appendix~\ref{app:md_results} $r$: Pearson correlation; $\rho$: Spearman correlation.}
    \label{fig:md_results}
\end{figure*}

\textbf{Diversifying solved structures.} Although we have so far focused on generating protein ensembles without the use of experimental structures, there is substantial scientific interest in obtaining ensembles for specific \emph{solved} structures, often via molecular dynamics simulation \citep{hollingsworth2018molecular}. To investigate the utility of our method in this application setting, we repeat all experiments by providing the structure which initialized the ATLAS simulations to the +Templates version of Alpha\textsc{Flow}-MD. As expected, the resulting ensembles improve---sometimes substantially---in their similarity to the ground truth MD ensemble, vastly surpasssing the performance of MSA subsampling. However, in these settings, the appropriate baseline is \emph{replicate simulations} (provided in ATLAS) starting from the same structure rather than MSA subsampling. Since MD is taken to be the ground truth but is expensive to run to convergence, we investigate if Alpha\textsc{Flow} provides better results for an equivalent \emph{limited computational budget}, e.g., in terms of GPU-hrs. To emulate these budgets, we reduce the number of samples drawn from Alpha\textsc{Flow} (from 250 to as few as 4) and the length of the MD trajectory (100 ns--160 ps). As shown in Figure~\ref{fig:md-convergence}, the Alpha\textsc{Flow} ensembles retain much of their quality with up to a 10x reduction in samples, while the MD trajectories require much longer to converge to or surpass the same quality. The distilled Alpha\textsc{Flow} model, despite converging to a lower level of performance, provides an even greater improvement for short timescales by providing 10x as many samples for the same runtime. Thus, as measured by these metrics, Alpha\textsc{Flow} provides a more efficient means to study the thermodynamic fluctuations of existing solved structures than short MD simulations and holds promise for large-scale diversification of the PDB.

\section{Conclusion}
We have presented Alpha\textsc{Flow} and ESM\textsc{Flow}, which combine AlphaFold and ESMFold with flow-matching towards the goal of sampling protein ensembles. Compared to existing approaches for obtaining multiple structure predictions, our method goes beyond \emph{inference-time} input modifications and develops a more principled \emph{training-time} approach to modeling structural diversity. Comprehensive experimental results demonstrate the utility and performance of our method in predicting precise and diverse PDB structures and replicating distributions and properties of MD ensembles, both with and without initial experimental structures. We anticipate these capabilities to have broad and exciting applications for structure biology. Further, with the increasing availability of high-resolution cryo-EM data \citep{kuhlbrandt2014resolution} and algorithms for resolving their structural heterogeneity \citep{zhong2021cryodrgn}, we anticipate the paradigm of generative training of AlphaFold and ESMFold to have further applications beyond the  settings considered here. 

\clearpage

\section*{Acknowledgements}
We thank Ruochi Zhang, Hannes St{\"a}rk, Samuel Sledzieski, Soojung Yang, Jason Yim, Rachel Wu, Hannah Wayment-Steele, Umesh Padia, Sergey Ovchinnikov, Andrew Campbell, Gabriele Corso, Jeremy Wohlwend, Mateo Reveiz, Martin V{\"o}gele, Daniel Richman, Jessica Karaguesian, Alex Powers, and anonymous ICML reviewers for helpful feedback and discussions. 

This work was supported by the National Institute of General Medical Sciences of the National Institutes of Health under award number 1R35GM141861-01 and the U.S. Department of Energy, Office of Science, Office of Advanced Scientific Computing Research, Department of Energy Computational Science Graduate Fellowship under Award Number DE-SC0022158. We acknowledge support from NSF Expeditions grant (award 1918839: Collaborative Research: Understanding the World Through Code), the Machine Learning for Pharmaceutical Discovery and Synthesis (MLPDS) consortium, the Abdul Latif Jameel Clinic for Machine Learning in Health, the DTRA Discovery of Medical Countermeasures Against New and Emerging (DOMANE) threats program,  and the DARPA Accelerated Molecular Discovery program. This research used resources of the National Energy Research Scientific Computing Center (NERSC), a Department of Energy Office of Science User Facility using NERSC awards ASCR-ERCAP0027302 and ASCR-ERCAP0027818.

\bibliography{references}
\bibliographystyle{icml2024}

\newpage
\appendix
\onecolumn

\section{Method Details}\label{app:algorithms}

\subsection{Input Embedding Module} \label{app:input_embedder}
Algorithm~\ref{alg:embedding} outlines the architecture of the input embedding module which we attach to AlphaFold and ESMFold to form Alpha\textsc{Flow} and ESM\textsc{Flow}, respectively. The output of the module is added to the input to the Evoformer or folding trunk. The various subroutines are as defined in AlphaFold \citep{jumper2021highly}, whereas the Gaussian Fourier time embeddings are as previously used in \citet{song2021score, tancik2020fourier}. For brevity, we have omitted droupout layers. 

\begin{algorithm}[H]
\caption{\textsc{InputEmbedding}}\label{alg:embedding}
\begin{algorithmic}
\STATE \textbf{Input:} {Beta carbon coordinates $\bfx \in \mathbb{R}^{N \times 3}$, time $t\in[0,1]$}
\STATE \textbf{Output:} {Input pair embedding $\bfz \in \mathbb{R}^{N \times N \times 64}$}
\STATE $\bfz_{ij} \gets \lVert \bfx_i - \bfx_j \rVert $ 
\STATE $\bfz_{ij} \gets \Bin(\bfz_{ij}, {\min}=3.25 \text{ \AA}, {\max}= 50.75 \text{ \AA},  N_\text{bins}=39)$ 
\STATE $\bfz_{ij} \gets \Linear(\OneHot(\bfz_{ij}))$ 
\FOR{$l \gets 1$ to $N_\text{blocks} = 4$}
\STATE    $\{\bfz\}_{ij} \pluseq \TriangleAttentionStartingNode({\bfz_{ij}}, c = 64, N_\text{head} = 4)$   
\STATE    $\{\bfz\}_{ij} \pluseq \TriangleAttentionEndingNode({\bfz_{ij}}, c = 64, N_\text{head} = 4))$ 
\STATE    $\{\bfz\}_{ij} \pluseq \TriangleMultiplicationOutgoing({\bfz_{ij}}, c = 64)$   
\STATE    $\{\bfz\}_{ij} \pluseq \TriangleMultiplicationIncoming({\bfz_{ij}}, c = 64)$ 
\STATE    $\{\bfz\}_{ij} \pluseq \PairTransition({\bfz_{ij}}, n = 2)$ 
\ENDFOR
\STATE $\bfz_{ij} \pluseq \Linear(\GaussianFourierEmbedding(t, d=256))$
\end{algorithmic}
\end{algorithm}

\subsection{Flow Matching on Protein Ensembles} \label{app:flow_matching}

In this subsection, we describe how the final training and inferences Algorithms~\ref{alg:training} and \ref{alg:inference} are obtained, starting from the Euclidean flow matching procedure from a harmonic prior provided in Section~\ref{sec:flow_matching}, Equations~\ref{eq:fm_start}--\ref{eq:fm_end}. We note that other diffusion or flow matching formulations are also possible and  leave further exploration of this design space to future work.

\paragraph{Unsuitability of MSE Loss} In standard flow matching over $\mathbb{R}^{3N}$, the denoising network $\hat\bfx(\bfx, t; \theta)$ is designed to approximate ${\hat \bfx}_1(\bfx, t;\theta) \approx \mathbb{E}_{\bfx_1 \sim p_t(\bfx_1 \mid \bfx)}[\bfx_1]$, which gives rise to the MSE training objective
\begin{equation}\begin{aligned}
    \mathcal{L}_t(\theta) &= \mathbb{E}_{\bfx_1 \sim p_\text{data}, \bfx \sim p_t(\bfx \mid \bfx_1)}\left[\lVert \hat\bfx_1(\bfx, t; \theta) - \bfx_1 \rVert^2 \right] \\
    &= \mathbb{E}_{\bfx_1 \sim p_\text{data}, \bfx_0 \sim q}\left[\lVert \hat\bfx_1(\bfx, t; \theta) - \bfx_1 \rVert^2 \right]
\end{aligned}\end{equation}
where $\bfx = (1-t)\cdot \bfx_0 + t \cdot \bfx_1$, for each time $t \in [0, 1]$. The harmomic prior density $q$ and the data distribution $p_\text{data}$ are $SE(3)$-invariant (technically $SO(3)$-invariant after centering; see for example \citet{yim2023se}). This means that, for each training pair $(\bfx_0, \bfx_1)$, there is a corresponding uniform density over $R \in SO(3)$ supplying examples $(R.\bfx_0, R.\bfx_1)$:
\begin{equation}
    \mathcal{L}_t(\theta) = \mathbb{E}_{\bfx_1 \sim p_\text{data}, \bfx_0 \sim q}\left[\int_{SO(3)} \lVert \hat\bfx_1(R.\bfx, t; \theta) - R.\bfx_1 \rVert^2 \, dR\right] 
\end{equation}
However, because the input embedding takes only a distogram of $\bfx$, the denoising model $\hat\bfx_1$, i.e., AlphaFold or ESMFold, is $SE(3)$-invariant, meaning that
\begin{equation}
    \hat\bfx_1(\bfx, t; \theta) = \hat\bfx_1(R.\bfx, t; \theta)
\end{equation}
for any $R \in SO(3) \subset SE(3)$. Hence, the denoising network is tasked with predicting $R.\bfx_1$ despite having no access to $R$. This is impossible and would lead to the network to degenerately predict $\hat \bfx_1 = \boldsymbol{0}$, showing that the MSE loss (or more broadly any non-$SE(3)$-invariant) loss is unsuitable with a $SE(3)$-invariant denoising network.

\paragraph{Flow Matching on the Quotient Space} To resolve the issue that AlphaFold and ESMFold are insensitive to $SE(3)$ transformations of the input, we consider flow matching over the quotient space $\mathbb{R}^{3N}/SE(3)$, such that inputs related by $SE(3)$ transformations are now defined to be identical. This quotient space, when defined with suitable care \citep{diepeveen2023riemannian}, gives a non-Euclidean, Riemannian manifold. The harmonic prior and data distribution can be straightforwardly projected to this space by taking the $SE(3)$ equivalency classes of each data point. The theory of flow matching over Riemmanian manifolds was developed by \citet{chen2023riemannian} and closely follows standard flow matching, except the conditional vector fields and the learned marginal vector fields are elements of the tangent space:
\begin{equation}
    u_t(\bfx \mid \bfx_1) \in T_\bfx\mathcal{M}, \quad
    {\hat v}(\bfx, t; \theta) := \mathbb{E}_{\bfx_1 \sim p_t(\bfx_1 \mid \bfx)}[u_t(\bfx \mid \bfx_1)] \in T_\bfx\mathcal{M}
\end{equation}
As in the Euclidean case, to develop a flow matching process, we require a conditional probability path and a corresponding conditional vector field. \citet{chen2023riemannian} propose to generalize the CondOT probability path by defining the interpolant $\psi_t(\bfx_0 \mid \bfx_1)$ to be the geodesic from $\bfx_0$ to $\bfx_1$, and then specifying $p_t(\bfx \mid \bfx_0)$ via 
\begin{equation}
    \bfx \mid \bfx_1 = \psi_t(\bfx_0 \mid \bfx_1), \quad \bfx_0 \sim q(\bfx_0)
\end{equation}
and the associated conditional vector field as
\begin{equation}
    u_t(\bfx \mid \bfx_1) = \frac{d}{dt}\psi_t(\bfx_0 \mid \bfx_1)
\end{equation}
Once the marginal vector field is learned, inference is performed by integrating the corresponding ODE over the manifold. To use this framework with protein structures and AlphaFold or ESMFold as the flow model, we make the following tweaks:

(1) We construct the interpolation between two elements in the quotient space $\mathbb{R}^{3N}/SE(3)$ to be given by RMSD alignment in the ambient space $\mathbb{R}^{3N}$, followed by linear interpolation in ambient space. Thus, as employed in Algorithm~\ref{alg:training}, the conditional probability path is sampled via
\begin{equation}\begin{aligned}
    \bfx_0 &\sim q(\bfx_0) \\
    \bfx_0 &\gets \RMSDAlign(\bfx_0, \bfx_1) \\
    \bfx \mid \bfx_1 &= (1-t) \cdot \bfx_0 + t \cdot \bfx_1
\end{aligned}\end{equation}
(2) Similar to the Euclidean case, we consider a reparameterization (cf. Equation~\ref{eq:reparameterization}) which allows a denoising model $\hat\bfx_1(\bfx, t; \theta)$ such as AlphaFold or ESMFold to give the direction of the learned marginal flow:
\begin{equation}
    \hat v(\bfx, t; \theta) = \frac{\log_\bfx \hat\bfx_1(\bfx, t; \theta)}{1-t}
\end{equation}
where the logarithmic map gives the direction of the interpolation connecting $\bfx$ to $\hat\bfx_1(\bfx, t; \theta)$ (discussed next). Unlike the Euclidean case, however, this expression does not provide a simple training objective for $\bfx_1$ in terms of a denoising loss. This is because flow matching requires minimizing error in the tangent space, which may not be easily related to distances on the manifold. Nevertheless, we posit that a model which minimizes denoising error can do a good job of implicitly learning the vector field. Thus, for some choice of distance metric $d$ over the manifold, we aim to learn the so-called Fr\'echet mean of the clean data distribution conditioned on noisy data:
\begin{align}\label{eq:frechet_mean}
     \hat\bfx_1(\bfx, t; \theta) \approx \arg\min_{\hat\bfx\in\mathcal{M}} \mathbb{E}_{\bfx_1 \sim p_t(\bfx_1 \mid \bfx)}\left[d^2(\bfx_1, \hat\bfx_1)\right]
\end{align}
As a sanity check, note that when $\mathcal{M}$ is a Euclidean space and $d$ is Euclidean distance, $d^2$ reduces to the usual MSE denoising loss whose minimizer is the conditional expectation of $p_t(\bfx_1 \mid \bfx)$, in agreement with Equation~\ref{eq:x1_prediction}.

(3) At inference time, in lieu of repeatedly evaluating the logarithmic map and integrating the vector field with the exponential map, we observe that such a procedure amounts to moving along the interpolant connecting $\bfx$ to $\hat\bfx_1$:
\begin{equation}
    \exp_\bfx\left[\hat v(\bfx, t; \theta) \, dt\right] = \exp_\bfx \left[\frac{dt}{1-t} \log_\bfx \hat\bfx_1(\bfx, t; \theta)\right]
\end{equation}
i.e., a fraction $dt/(1-t)$ towards $\hat\bfx_1$. Hence, we take an integration step at inference-time via RMSD alignment followed by linear interpolation in ambient space, as executed in Algorithm~\ref{alg:inference}. 

\paragraph{FAPE and All-Atom Structure} As defined in Section~\ref{sec:flow_matching}, our flow matching framework operates over residue-level structures; specifically, over C$\beta$ coordinates $\bfx \in \mathbb{R}^{3N}$. However, the FAPE loss is defined over structures also containing (1) all-atom positions and (2) residue frames, and indeed we continue to supervise these outputs to ensure that Alpha\textsc{Flow} and ESM\textsc{Flow} produce meaningful all-atom structures. To reconcile these views, let $\mathcal{S}$ denote an all-atom structure, let $\left[\,\cdot\,\right]_{\text{C}\beta}$ be the operator that extracts the C$\beta$ coordinates, and denote the denoising model as $\hat {\mathcal{S}}(\bfx, t;\theta)$. Most of training and inference proceeds as if all structures were passed through the $\left[\,\cdot\,\right]_{\text{C}\beta}$ operator: training points are sampled via $\bfx_1 = \left[\mathcal{S}\right]_{\text{C}\beta}$ before noisy interpolation; and inference proceeds by parameterizing the C$\beta$ denoising model as
\begin{equation}
    \hat\bfx_1(\bfx, t; \theta) = \left[\hat{\mathcal{S}}(\bfx, t;\theta)\right]_{\text{C}\beta}
\end{equation}
However, this extraction is \emph{not} applied to compute the denoising loss---neither to the sampled data nor the prediction. Instead, the denoising model is trained to approximate (cf. Equation~\ref{eq:frechet_mean}):
\begin{equation}
    \hat{\mathcal{S}}(\bfx, t; \theta) \approx \arg\min_{\hat{\mathcal{S}}} \mathbb{E}_{\mathcal{S} \mid \bfx}\left[\FAPE^2(\mathcal{S}, \hat{\mathcal{S}})\right] 
\end{equation}
and thus the reparameterized C$\beta$ denoising model becomes
\begin{equation}
    \hat\bfx_1(\bfx, t; \theta) \approx \left[\arg\min_{\hat{\mathcal{S}}} \mathbb{E}_{\mathcal{S} \mid \bfx}\left[\FAPE^2(\mathcal{S}, \hat{\mathcal{S}})\right] \right]_{\text{C}\beta}
\end{equation}
Colloquially, this means that the denoised C$\beta$ structure (towards which we interpolate at inference time) is the C$\beta$ \emph{part of the best all-atom prediction}, rather than the best C$\beta$ prediction. In the final inference step, rather than extracting $\hat\bfx_1$ from $\hat{\mathcal{S}}$ and interpolating the rest of the way towards it, we simply return the all-atom structure $\hat{\mathcal{S}}$. However, the model is predicting the denoised all-atom structure from the C$\beta$ structure alone, and there is no iterative refinement of the non C$\beta$ components. Hence, our model is best thought of as a generative model over C$\beta$ positions \emph{only}, which additionally fills in the all-atom information to minimize the FAPE loss conditioned on the input C$\beta$ positions.

\subsection{Comparison with Harmonic Diffusion} \label{app:harmonic_diff}
In harmonic diffusion, as in flow matching, a conditional probability path $p(\bfx_t \mid \bfx_0)$ represents a noising process for the data point $\bfx_0$ ($t=0$ for the data by diffusion convention). Unlike flow matching, the path is given by the transition (or perturbation) kernel of a (Markovian) diffusion process rather than interpolation with the noise. The stationary distribution of the diffusion is the noisy prior by construction; however, the probability path converges to this prior only in the infinite-time limit. Instead, the maximum time is chosen such that the KL-divergence between the $p_{t\mid 0}$ and the stationary distribution is acceptably low. Unfortunately, in harmonic diffusion:
\begin{equation}
    D_\text{KL}(p_{t\mid 0} || p_\infty) = \sum_{i=1}^{3n} \left[e^{-\lambda_i t} \left(E_i - \frac{1}{2}\right) - \frac{1}{2}\log \left(1 - e^{-\lambda_i t}\right)\right]
\end{equation}
where $E_i$ is the (roughly constant) amount of energy in the $i$th mode (Equation~3 in \citet{jing2023eigenfold}). That is, the rate of convergence not only depends on the number of dimensions, but---more problematically---the smallest eigenvalue $\lambda_i$ of the diffusion drift matrix, which becomes smaller for larger proteins. Hence, it becomes tricky to train a time-conditioned denoising model for proteins of arbitrary size. In the case of Alpha\textsc{Flow} and ESM\textsc{Flow} trained on crops of 256, the model would not be able to denoise longer proteins from an intermediate state at which the crops have converged to noise, but the entire protein has not---such states have never been seen during training. Our flow matching framework instead imposes the noisy prior as a boundary condition at the same $t=0$ \emph{for all protein lengths and crops}, avoiding this issue.

While the fixed convergence time is a desirable quality, our flow matching framework---at least as defined in Equations~\ref{eq:fm_start}--\ref{eq:fm_end}---satisfies an even stronger property, which we call \emph{crop invariance} (Proposition~\ref{prop:crop_invariance}). Colloquially, this means that the marginal distribution of a crop of length $M$ at time $t$ is the same as if it were noised independently as an intact sequence of length $M$. This property ensures the noisy distributions over isolated crops seen at training time are exactly the same as those seen at inference time, when those crops are embedded in full-size proteins. 

\begin{proposition}\label{prop:crop_invariance}
Let $\bfx_1 \in \mathbb{R}^N$ and $\bfx_1 \verb|[i:i+M]| \in \mathbb{R}^M$ be a crop of $\bfx_1$ of length $M \le N$ and define $p_t^{(M)}, p_t^{(N)}$ to be the conditional probability paths in dimensionalities $N, M$. Then for any $t, \tilde\bfx \in \mathbb{R}^M$, $p_t^{(N)}(\bfx \verb|[i:i+M]| = \tilde \bfx \mid \bfx_1)$ is equal to $p_t^{(M)}(\bfx = \tilde \bfx \mid \bfx_1 \verb|[i:i+M]|)$.
\end{proposition}
\begin{proof}
Our key claim is that for time $t=0$, i.e. in the noise distribution, the density $q^{(N)}(\bfx \verb|[i:i+M]| = \tilde \bfx)$ is equivalent to $q^{(M)}(\bfx = \tilde \bfx)$. The former amounts to marginalizing the density $q^{(N)}(\bfx)$ over the non-crop variables. For simplicity, we proceed with $i=0$; the more general case is very similar:
\begin{align*}
    q^{(N)}\left(\bfx_{[0,M)} = \tilde \bfx\right) &= \int q^{(N)}\left(\bfx_{[0,M)} = \tilde\bfx, \bfx_{[M, N)}\right) \; d\bfx_{[M, N)}\\
    &\propto \int \exp\left[-\frac{\alpha}{2} \left[\sum_{j=0}^{M-2}\lVert \tilde\bfx_j - \tilde \bfx_{j+1} \rVert^2 + \lVert \tilde \bfx_{M-1} - \bfx_M\rVert^2 + \sum_{j=M}^{N-2} \lVert \bfx_j - \bfx_{j+1}\rVert^2\right]\right] d\bfx_{[M, N)} \\
    &= \exp\left[-\frac{\alpha}{2}\sum_{j=0}^{M-2}\lVert \tilde\bfx_j - \tilde \bfx_{j+1} \rVert^2 \right] \underbrace{\int \exp\left[-\frac{\alpha}{2} \left[\lVert \tilde \bfx_{M-1} - \bfx_M\rVert^2 + \sum_{j=M}^{N-2} \lVert \bfx_j - \bfx_{j+1}\rVert^2\right]\right] d\bfx_{[M, N)}}_\text{constant} \\
    &\propto q^{(M)}(\bfx = \tilde \bfx)
\end{align*}
where the constant is an offset Gaussian integral. This equivalence means that---up to some global translation---sampling noise of dimension $N$ and then cropping to length $M$ is equivalent to sampling noise of dimension $M$. Then, notice that linear interpolation of full structures implies linear interpolations of crops:
\begin{equation*}
\bfx = (1-t) \cdot \bfx_0 + t \cdot \bfx_1 \quad \implies \quad \bfx\verb|[i:i+M]| = (1-t) \cdot \bfx_0\verb|[i:i+M]| + t \cdot \bfx_1\verb|[i:i+M]|
\end{equation*}
Thus, the sampling procedure for $p_t^{(N)}(\bfx \verb|[i:i+M]| = \tilde \bfx \mid \bfx_1)$---which is to interpolate $N$-dimensional noise and data and then crop to $M$ dimensions---is the same as the sampling procedure for $p_t^{(M)}(\bfx = \tilde \bfx \mid \bfx_1 \verb|[i:i+M]|)$---which is to first crop the data and noise to $M$ dimensions and then interpolate.
\end{proof}
We note that crop invariance no longer holds in the final form of flow matching that we use in Algorithms~\ref{alg:training} and \ref{alg:inference} and describe in Appendix~\ref{app:flow_matching} due to the RMSD alignment step. Nevertheless, we posit that initial preservation of distributional alignment helps with generalization to proteins of unseen large sizes at inference time.

The second advantage of our flow matching framework over harmonic diffusion is in the treatment of missing residues. In harmonic diffusion, the perturbation kernel $p(\bfx_t \mid \bfx_0)$ is a Gaussian whose mean is given by $\boldsymbol{\mu} = e^{-t\mathbf{H}/2}\bfx_0$, where $\mathbf{H}$ is the drift matrix. This matrix exponential is far from diagonal, meaning that each entry of $\boldsymbol{\mu}$ is dependent on all initial entries of $\bfx_0$. Hence, if there are missing coordinates in $\bfx_0$, they must be imputed in order to sample $p(\bfx_t \mid \bfx_0)$. In contrast, in our flow matching framework, each coordinate in $\bfx$ at time $t$ is a linear combination of only the same-index coordinates in $\bfx_0$ and $\bfx_1$. Hence, we can simply omit the missing residues in the RMSD alignment and the subsequent interpolation.

\section{Experimental Details}\label{app:exp_details}

\subsection{Training and Inference} \label{app:training_inference}

\paragraph{Training} We use OpenFold \citep{ahdritz2022openfold} to train Alpha\textsc{Flow} and ESM\textsc{Flow}, as it closely follows the training best practices described in AlphaFold \citep{jumper2021highly}. However, because the OpenFold weights for AlphaFold were trained with a much later cutoff date, we instead initialize with the original CASP14 weights from DeepMind (version 1). For PDB training data, we use a January 2023 snapshot of the PDB and apply 40\% clustering with MMSeqs2 \citep{steinegger2017mmseqs2}. We train with crops of size 256, batch size of 64, no recycling, and no templates. Alpha\textsc{Flow} is trained on the full set of auxiliary losses, except the structural violation loss and with the FAPE loss squared. ESM\textsc{Flow} is trained on the FAPE, pLDDT, distogram, and supervised $\chi$ losses. To maintain precision in the initial prediction, we set $t=0$ and omit the noisy input in 20\% of training examples. Training progress is monitored via the precision and diversity on a validation set of 183 CAMEO targets deposited Aug--Oct 2022, following \citet{jing2023eigenfold}. To fine-tune on MD ensembles, we resume from the selected checkpoints from the PDB training. All the training settings remain unchanged, except the targets are sampled uniformly at random (with a random conformation), the batch size is set to 8, and $t=0$ is set 10\% of the time. Training progress is monitored via the loss on the validation split.

\paragraph{Training Cost} All training is done on a machine with 8x NVIDIA A100 GPUs and 2x Intel Xeon(R) Gold 6258R processors, with the total training cost shown in Table~\ref{tab:training_cost}.

\begin{table}[H]
    \centering
    \caption{\textbf{Alpha\textsc{Flow} and ESM\textsc{Flow} training cost}}
    \label{tab:training_cost}
    \begin{tabular}{clccc}
    \toprule
    & & \makecell{Total\\hours} & \makecell{Train\\examples} & \makecell{Secs per\\training pass}\\
    \midrule
    \multirow{6}{*}{Alpha\textsc{Flow}}
    & PDB & 267 & 1.28M & 5.8 \\
    & PDB distillation & 105 & 160k & 17.4 \\
    & MD & 11 & 43k & 6.2 \\
    & MD distillation & 28 & 38k & 17.4 \\
    & MD+Templates & 9 & 38k & 6.3 \\
    & MD+Templates distillation & 39 & 51k & 18.0 \\
    \midrule
    \multirow{6}{*}{ESM\textsc{Flow}}
    & PDB & 104 & 720k & 4.2 \\
    & PDB distillation & 37 & 64k & 11.9 \\
    & MD & 5 & 27k & 4.6 \\
    & MD distillation & 34 & 51k & 12.0 \\
    & MD+Templates & 9 & 51k & 4.7 \\
    & MD+Templates distillation & 23 & 38k & 12.5 \\
    \bottomrule
    \end{tabular}
\end{table}

\paragraph{Inference} We run Alpha\textsc{Flow} and ESM\textsc{Flow} with 10 steps by default, evenly spaced from $t=0$ to $t=1$, where the first prediction is performed with no noisy input. However, by merging the first $K>1$ steps, we can reduce the variance of the sampled distribution and increase precision, analogous to increasing MSA depth. This is because---after the initial large step to $t=0.1K$---we are effectively starting the flow from a modified intermediate marginal $p_t(\bfx)$ which differs from the $p_t(\bfx)$ that would arise from properly following the flow:
\begin{align}
    \bfx &= (1-t) \cdot \bfx_0 + t \cdot \bfx_1, \quad \bfx_0 \sim q(\bfx_0), \bfx_1 \sim p_\text{data}(\bfx_1) &\text{(original)} \\
    \bfx &= (1-t) \cdot \bfx_0 + t \cdot \mathbb{E}_{p_\text{data}}[\bfx_1], \quad \bfx_0 \sim q(\bfx_0) &\text{(modified)}
\end{align}
i.e., by stepping directly to intermediate time $t$, we interpolate towards the \emph{initial} $\hat\bfx_1$ prediction, which is a single point estimate of the unconditional expectation, rather than the full distribution $p_\text{data}(\bfx_1)$. We omit recycling for all methods following \citet{del2022sampling}. Note that, by default, AlphaFold accepts a maximum MSA depth equivalent to subsampling with depth 1024, and exhibits a small level of diversity; on the other hand, ESMFold is completely deterministic. For Alpha\textsc{Flow} PDB experiments, we resample the MSA (with depth 1024) for each new sample, but not for each inference step. At inference time, MSAs for all sequences are computed with the ColabFold MMSeqs pipeline \citep{porter2023colabfold}.

\paragraph{Self-conditioning} Although we do not use recycling \emph{per se} for either our methods or the baselines, we employ \emph{self-conditioning} \citep{chen2022analog, stark2023harmonic} in the PDB experiments to increase the precision of Alpha\textsc{Flow}. In particular, at training time, 50\% of supervised forward passes are provided the (gradient-detached) outputs from an initial forward pass of the model; we reuse the recycling embedder of AlphaFold to embed these outputs. At inference time, every forward pass after the first is provided the outputs of the previous forward pass. Note that unlike \citet{stark2023harmonic}, we self-condition with the full set of model output states, i.e., including pair embeddings, rather than just the output $\hat\bfx_1$ prediction. Self-conditioning is omitted for distillation training and for MD training and inference. Finally, although we also trained ESM\textsc{Flow} with self-conditioning, we did not observe any improvements and report results without it.

\paragraph{Distillation} Because the inference process is deterministic except for the initial noisy sample, it defines a map from the noisy distribution to the data distribution. We can aim to learn this map via a model that ingests the noisy sample and predicts the corresponding fully-denoised output in a single forward pass. To train such a model, for each training example (still a crop of 256), we run the full inference pipeline with the original flow model and set the result as the training target. All other training settings are kept the same and training performance is monitored the same way, except the batch size is always set to 8, and the concepts of sampling $t$, interpolating, and self-conditioning no longer apply. For Alpha\textsc{Flow} and ESM\textsc{Flow} on the PDB, we train for 160k and 64k training examples, respectively. For distilling the MD models, we start from the weights of the original Alpha\textsc{Flow}-MD and ESM\textsc{Flow}-MD and fine tune for 38k and 51k training examples, respectively.

\subsection{Datasets} \label{app:datasets}
\paragraph{PDB Test Set} To construct the test set of structurally heterogeneous recent proteins from the PDB, we follow \citet{ellaway2023identifying} and identify chains as representing the same protein if they map to the same segment in the same UniProt reference sequence. We use the SIFTS annotations database \citep{dana2019sifts} and its residue-level mappings from PDB chains to UniProt reference sequences to associate each deposited chain with a segment. Then, we cluster all segments with a Jaccard similarity threshold of 0.75 and complete linkage, with each resulting cluster regarded as a distinct protein, yielding 75k proteins. We collect all proteins which (1) are represented by 2--30 chains deposited after the AlphaFold training cutoff and no chains before the cutoff, (2) have lengths between 256--768 residues, (3) have at least two structural clusters when the chains are clustered with a threshold of 0.85 symmetrized lDDT-C$\alpha$ and complete linkage. From the resulting 563 proteins (represented by 2843 chains), we subsample 100 proteins (represented by 500 chains) to form the test set. At inference time, we run all models using the sequence given by the UniProt segment. The distribution of sequences lengths is shown in Figure~\ref{fig:hist_seqlen}.

\paragraph{MD Dataset} The ATLAS dataset \citep{vander2023atlas} consists of all-atom, explicit solvent MD simulations for 1390  non-membrane proteins, chosen as representatives for all eligible ECOD structural classes \citep{schaeffer2017ecod}. For each protein, 3 replicate simulations of length 100 ns are provided, each with 10k frames. To train and validate on these trajectories, we first generate MSAs for all 1390 ATLAS entries using the provided sequence and the ColabFold MMSeqs2 pipeline \citep{porter2023colabfold}. Then, for the train and validation sets, we extract 300 conformations to be randomly sampled in the training pipeline. The test split consists of all 84 targets whose corresponding PDB entries were deposited after May 1, 2019, minus the two targets with sequence length greater than 1024. The resulting distribution of sequences lengths is shown in Figure~\ref{fig:hist_seqlen}.

\begin{figure}[H]
    \centering
    \includegraphics[width=0.65\textwidth]{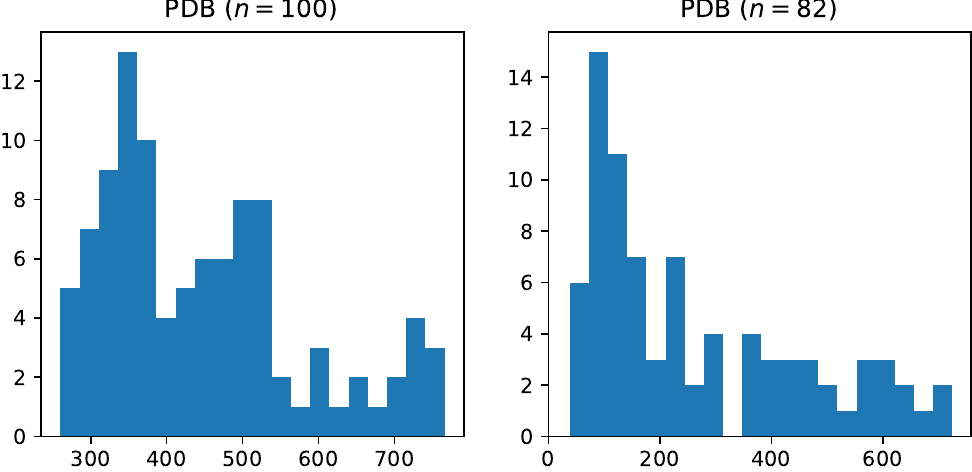}
    \caption{\textbf{Histogram of sequence lengths} in the PDB test set (\emph{left}) and the ATLAS test set (\emph{right}).}
    \label{fig:hist_seqlen}
\end{figure}

\subsection{Evaluation Procedures} \label{app:evaluation}

\paragraph{Symmetrized lDDT} In the PDB experiments, we often need to compute the similarity (or dissimilarity) between two structures which may not share identical sequences, and which may differ significantly in length---for example, between two PDB chains or between a PDB chain and a structure predicted from the UniProt reference sequence. To do so, we define the \emph{symmetrized lDDT} as a variant of lDDT-C$\alpha$ which is (as the name suggests) symmetric and robust to these discrepancies. We perform a pairwise alignment of the two sequences, and tabulate the C$\alpha$ pairs (identified by residue index only) which are within 15 \AA\ of each other in \emph{either structure}. Then, we score the fraction of these selected pairwise distances that are consistent within 0.5 \AA, 1 \AA, 2 \AA, and 4 \AA\ in the two structures. The symmetrized lDDT-C$\alpha$ is the mean of these four scores.

\paragraph{MD Evaluations} To compare a generated ensemble with the ground-truth MD ensemble, we first align both ensembles to the static all-atom structure that initialized the simulation (provided in the ATLAS download). We then perform all analyses using the Euclidean atomic coordinates in MDTraj \citep{mcgibbon2015mdtraj}. For most procedures, we subsample 1000 random MD frames to reduce the analysis runtime. To compute the RMWD, the Wasserstein distance between two 3D Gaussians is given by
\begin{equation}
    \mathcal{W}_2^2(\mathcal{N}(\mu_1, \Sigma_1), \mathcal{N}(\mu_2, \Sigma_2)) = \lVert \mu_1 - \mu_2 \rVert^2 + \Tr\left(\Sigma_1 + \Sigma_2 - 2(\Sigma_1\Sigma_2)^{1/2}\right)
\end{equation}
which reduces to Euclidean distance in the case of point masses. This squared distance decomposes into a translation term and a variance term; so the aggregate RMWD (Equation~\ref{eq:rmwd}) also decomposes as
\begin{equation}
    \RMWD^2(\mathcal{X}_1, \mathcal{X}_2) = \underbrace{\frac{1}{N} \sum_{i=1}^N \lVert \mu_{1,i} - \mu_{2,i} \rVert^2}_{(\text{translation contribution})^2} + \underbrace{\frac{1}{N} \sum_{i=1}^N \Tr\left(\Sigma_{1,i} + \Sigma_{2,i} - 2(\Sigma_{1,i}\Sigma_{2,i})^{1/2}\right)}_{(\text{variance contribution})^2}
\end{equation}
We report the translation contribution (which resembles RMSD) and variance contribution in Table~\ref{tab:md_results}. In the calculation of joint $\mathcal{W}_2$ distance, we first project to the PCA subspace because thermal fluctuations dominate in the full dimensional space and make the $\mathcal{W}_2$ metric unsuitable without an extremely large number of samples. While it is common to perform PCA using the MD reference ensemble alone, we note that doing so can obscure deviations of the predicted ensemble along the orthogonal degrees of freedom. Thus, we repeat the analysis using with the MD ensemble and the equally-weighted pooling of the MD and predicted ensembles. Finally, in the residue exposure analysis, we compute the solvent-accessible surface area (SASA) of each sidechain using the Shrake-Rupley algorithm and a probe radius of 2.8 \AA. Following \citet{porter2019cooperative}, we use a SASA threshold of 2.0 \AA$^2$ to distinguish buried and exposed residues. 

\paragraph{Comparison with Replicate MD} To compare the performance of our method with replicate MD simulations, we leverage the fact that ATLAS trajectories are provided in three replicates (100 ns and 10k frames each). In the main experiments, these three replicates are pooled to collectively represent the MD ensemble; however, such pooling would not be appropriate if one of these replicates is taken for comparison. Instead, we select the \emph{first} replicate for comparison and pool the latter two to represent the ground truth MD ensemble. We emulate different computational budgets by truncating the first trajectory to its first 4096, 2048, 1024, 512, 256, 128, 64, 32, and 16 frames before analysis, respectively representing simulation lengths of 40.96 ns, 20.48 ns, 10.24 ns, 5.12 ns, 2.56 ns, 1.28 ns, 640 ps, 320 ps, and 160 ps. When necessary, we subsample or replicate by the appropriate power of 2 to ensure all analyses operate on 256 frames (important for finite-sample Wasserstein distances). The computational cost in GPU-hrs is estimated by running 1 minute of MD for each protein on a single NVIDIA A100 GPU and noting the average performance in hrs/ns. (The average GPU utilization is 62\%, indicating efficient usage of resources.) For the Alpha\textsc{Flow} and ESM\textsc{Flow} ensembles, we first generate 250 samples as usual and also subsample 128, 64, 32, 16, 8, and 4 samples for analyses, duplicating by the appropriate power 2 to reach 256 ($\approx$ 250) samples. The runtime is provided as an average over all test proteins on a single NVIDIA A100 GPU. 

\paragraph{Comparison with Normal Mode Analysis} We also compare the performance of our method with normal mode analysis of the PDB protein structures using ProDy \citep{bakan2011prody}. We construct Gaussian Network Models (GNM) \citep{bahar1997direct} and Anisotropic Network Models (ANM) \citep{atilgan2001anisotropy} using the C$\alpha$ coordinates and draw 250 samples from each model, keeping all nondegenerate eigenvectors. We use $\Gamma = 0.15$ (adjusted from default to match the average MD RMSF) and default 10 \AA\ and 15 \AA\ cutoffs for GNM and ANM, respectively.

\newpage
\section{Additional Results}

\subsection{PDB Ensembles} \label{app:pdb_results}

Table~\ref{tab:pdb_results} provides precision, recall, and diversity results for the experiments on PDB ensembles, with a median taken over the 100 test set targets. For ESMFold and ESM\textsc{Flow}, the second set of results corresponds to the subset of targets released after the training cutoff of May 1, 2020 ($n=56$). Runtime measurements (per sample) are performed on a single A100 GPU and reported as a median over 100 targets. Figure~\ref{fig:pdb_pca} shows PCA of the true and generated ensembles for several selected targets to illustrate the degradation of the MSA subsampling ensembles. Figure~\ref{fig:pdb_conf_changes} highlights conformational changes observed in the PDB ensembles and correctly sampled by Alpha\textsc{Flow}. In both figures, the PCA is performed by first aligning all PDB sequences with the UniProt reference with ClustalW \citep{larkin2007clustal} and taking the C$\alpha$ positions of the common subset of aligned residues. The structures are then RMSD aligned to a randomly selected PDB structure and the PCA is performed on the resulting Euclidean coordinates. Sample weights are chosen so that the PDB structures account for half the loading, regardless of their number. Coordinates are converted to \AA\ RMSD units.

\begin{table}[h]
    \centering
    \caption{\textbf{Evaluation on PDB ensembles.}}
    \label{tab:pdb_results}
    \begin{tabular}{clcccc}
    \toprule
    & & Precison & Recall & Diversity & Runtime \\
    \midrule
    \multirow{4}{*}{Alpha\textsc{Flow}}
    & Full   & 0.810 & 0.801 & 0.185 & 69.6 \\
    & 5 steps  & 0.821 & 0.801 & 0.151 & 42.1 \\
    & 2 steps  & 0.839 & 0.811 & 0.082 & 21.3 \\
    & Distilled & 0.831 & 0.810 & 0.128 & 7.4 \\
    \midrule
    \multirow{5}{*}{\makecell{MSA\\subsampling}}
    & 512 & 0.849 & 0.823 & 0.026 & 5.5 \\
    & 256 & 0.844 & 0.818 & 0.044 & 4.2 \\
    & 128 & 0.835 & 0.806 & 0.053 & 3.9 \\
    & 64 & 0.795 & 0.784 & 0.088 & 3.6 \\
    & 48 & 0.757 & 0.760 & 0.125 & 3.5 \\
    \midrule
    AlphaFold & & 0.850 & 0.823 & 0.026 & 7.7 \\
    \midrule
    \multirow{4}{*}{ESM\textsc{Flow}}
    & Full & 0.777 / 0.777 & 0.777 / 0.765 & 0.210 / 0.213 & 30.4 \\
    & 5 steps  & 0.787 / 0.788 & 0.772 / 0.767 & 0.166 / 0.174 & 18.3 \\
    & 2 steps  & 0.795 / 0.797 & 0.774 / 0.760 & 0.100 / 0.102 & 9.2 \\
    & Distilled & 0.775 / 0.774 & 0.752 / 0.745 & 0.152	/ 0.152 & 3.1 \\
    \midrule
    ESMFold & & 0.806 / 0.809 & 0.764 / 0.761 & 0.000 & 3.2 \\
    \bottomrule
    \end{tabular}
\end{table}

\begin{figure}[h!]
    \centering
    \includegraphics[height=0.35\textwidth]{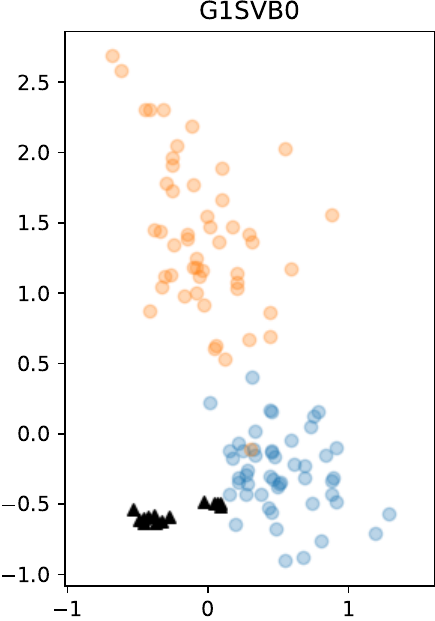}
    \includegraphics[height=0.35\textwidth]{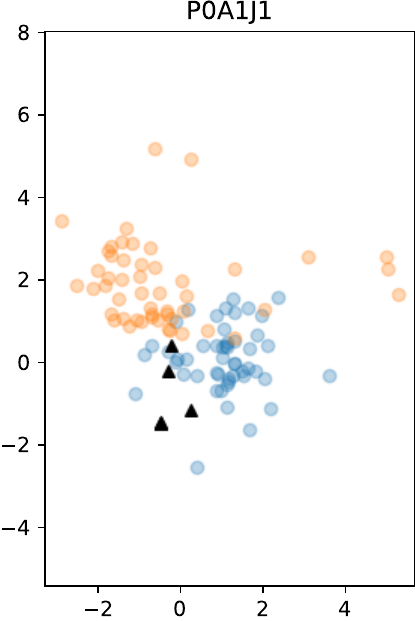}
    \includegraphics[height=0.35\textwidth]{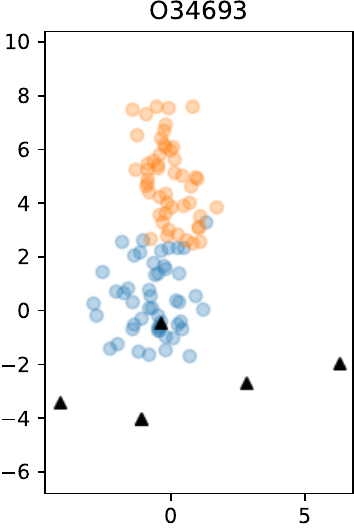}
    \includegraphics[height=0.35\textwidth]{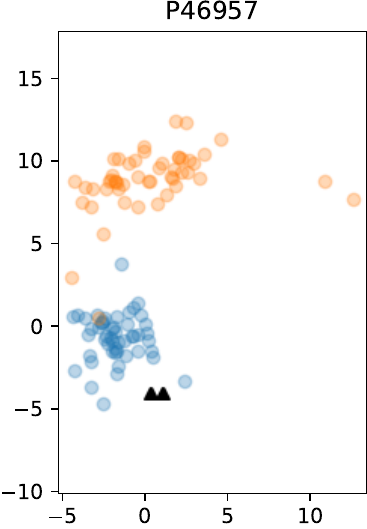}
    \caption{\textbf{PCA of PDB and predicted ensembles} from Alpha\textsc{Flow} (\textcolor[HTML]{1f77b4}{blue}) and MSA subsampling (depth 64) (\textcolor[HTML]{ff7f0e}{orange}), with PDB structures marked by $\blacktriangle$. The MSA subsampling ensembles have similar diversity as the Alpha\textsc{Flow} ensembles but drift away from the true structures.}
    \label{fig:pdb_pca}
\end{figure}

\clearpage

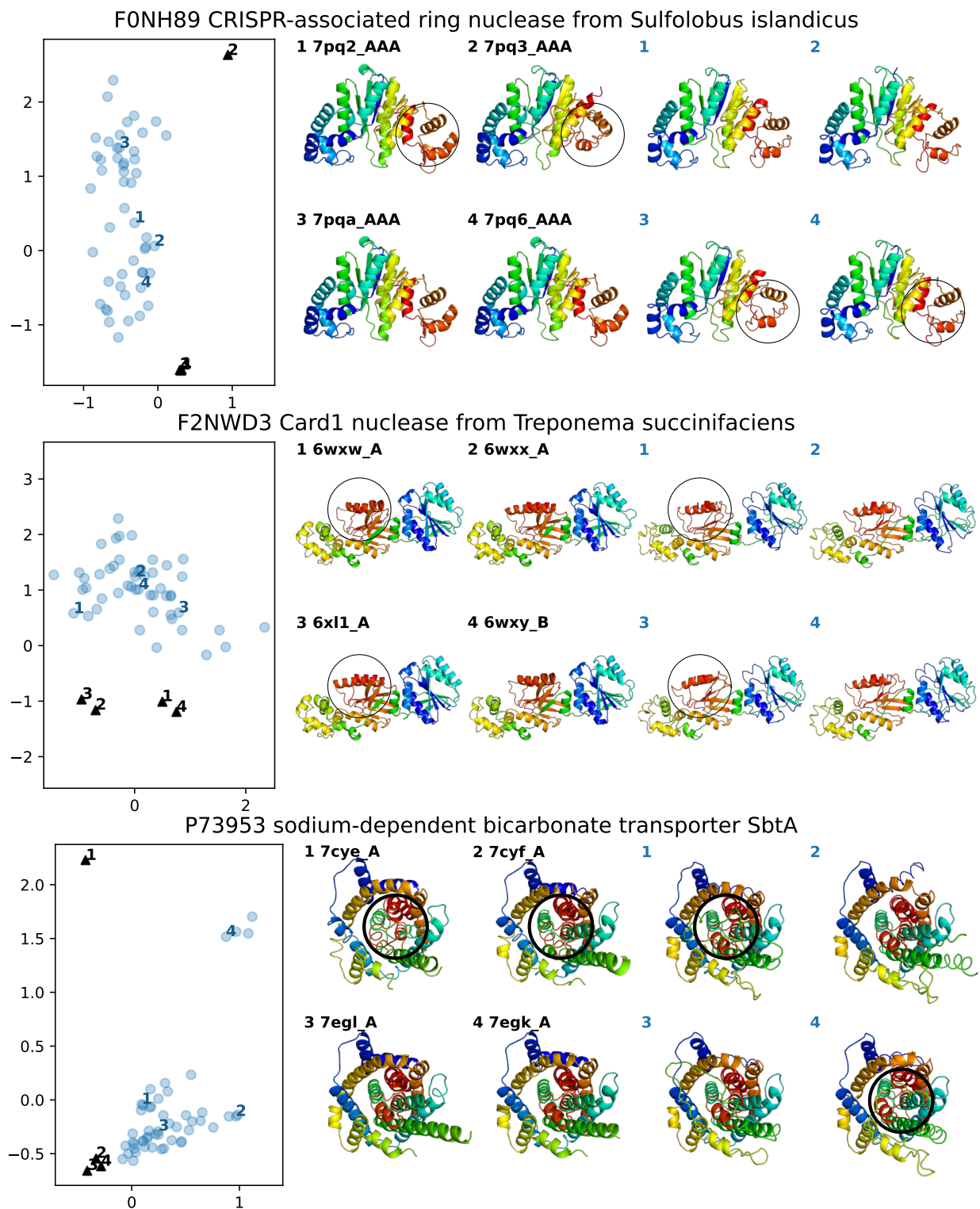
\begin{figure}[h]
    \centering
    \begin{tikzpicture}
    \node[anchor=south west,inner sep=0] (image) at (0,0) {\includegraphics[width=0.95\textwidth]{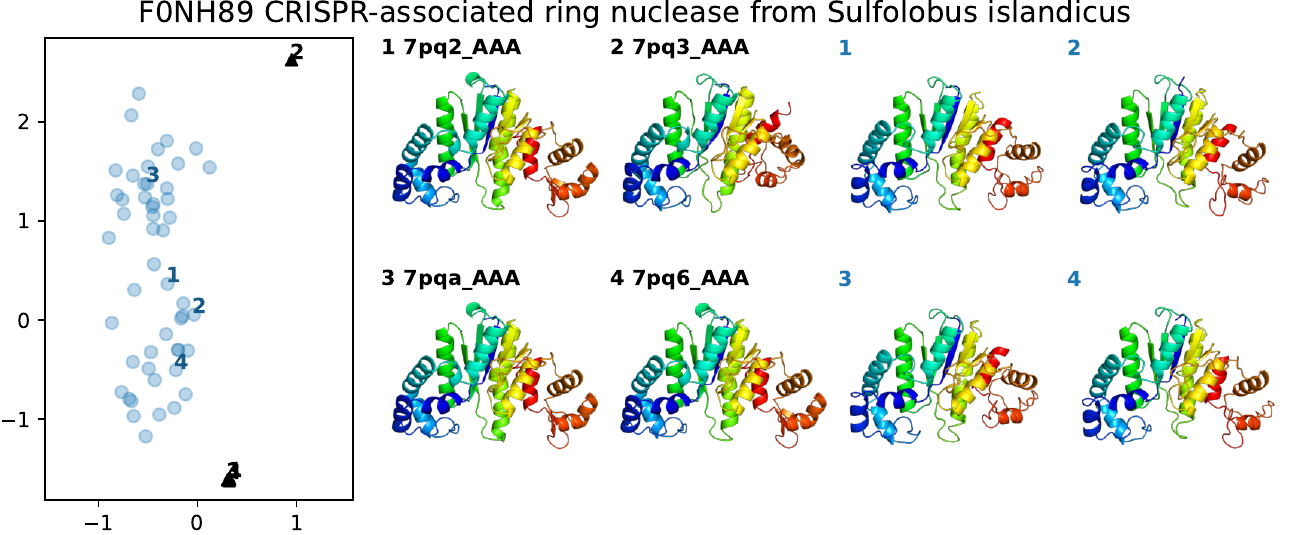}};
    \begin{scope}[x={(image.south east)},y={(image.north west)}]
        \draw[black] (0.43,0.69) circle [radius=15px];
        \draw[black] (0.60,0.69) circle [radius=15px];
        \draw[black] (0.78,0.25) circle [radius=15px];
        \draw[black] (0.95,0.25) circle [radius=15px];
    \end{scope}
    \end{tikzpicture}
    \begin{tikzpicture}
    \node[anchor=south west,inner sep=0] (image) at (0,0) {\includegraphics[width=0.95\textwidth]{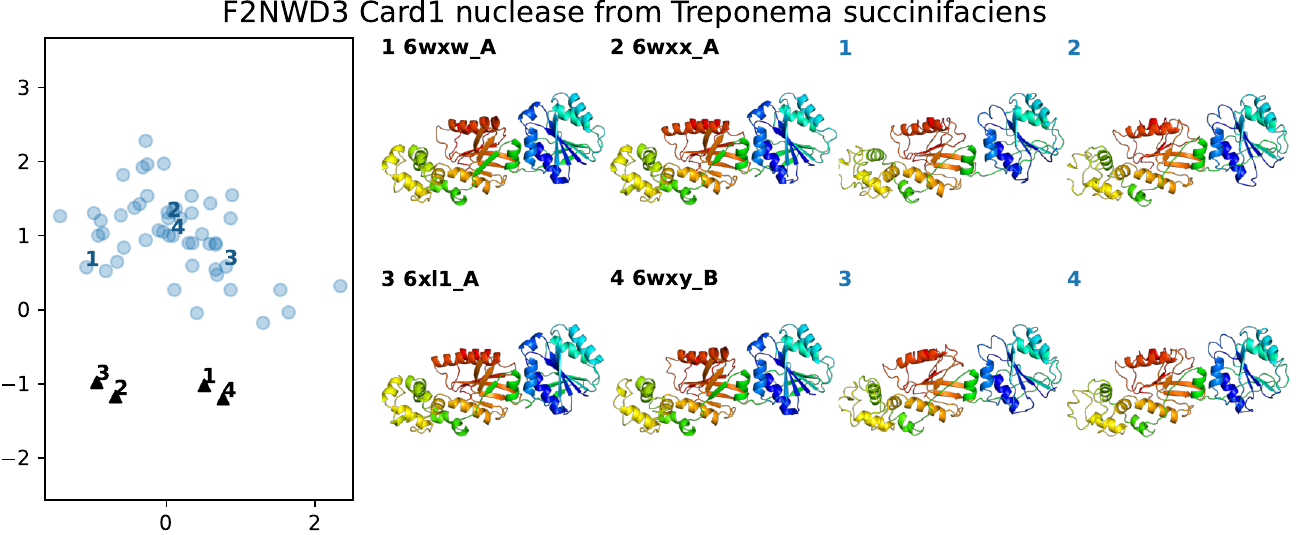}};
    \begin{scope}[x={(image.south east)},y={(image.north west)}]
        \draw[black] (0.36,0.76) circle [radius=15px];
        \draw[black] (0.36,0.32) circle [radius=15px];
        \draw[black] (0.71,0.76) circle [radius=15px];
        \draw[black] (0.71,0.32) circle [radius=15px];
    \end{scope}
    \end{tikzpicture}
    \begin{tikzpicture}
    \node[anchor=south west,inner sep=0] (image) at (0,0) {\includegraphics[width=0.95\textwidth]{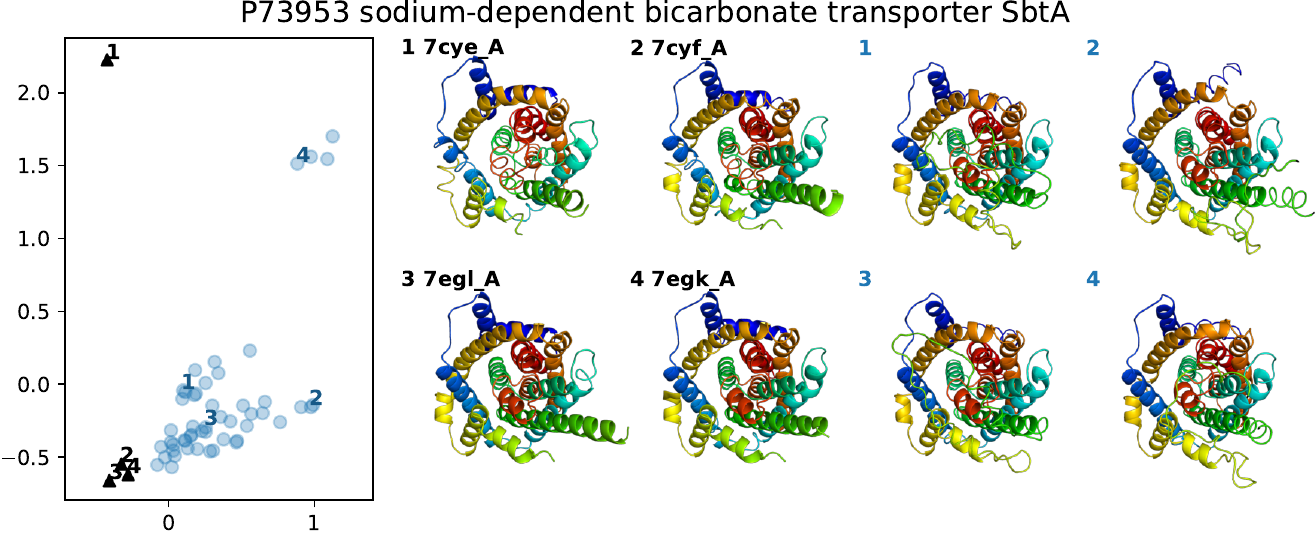}};
    \begin{scope}[x={(image.south east)},y={(image.north west)}]
        \draw[black,ultra thick] (0.4,0.72) circle [radius=15px];
        \draw[black,ultra thick] (0.57,0.72) circle [radius=15px];
        \draw[black,ultra thick] (0.74,0.72) circle [radius=15px];
        \draw[black,ultra thick] (0.92,0.28) circle [radius=15px];
    \end{scope}
    \end{tikzpicture}
    \caption{\textbf{PDB conformational changes} correctly sampled by Alpha\textsc{Flow}. For each UniProt ID, the PCA plot shows the complete set of PDB structures ($\blacktriangle$) and Alpha\textsc{Flow} samples (\textcolor[HTML]{1f77b4}{blue}). Four random PDB structures (\emph{left}) and four random Alpha\textsc{Flow} samples (\emph{right}) are visualized, where the numbers label the positions of the selected structures.}
    \label{fig:pdb_conf_changes}
\end{figure}

\clearpage
\subsection{MD Ensembles} \label{app:md_results}
Table~\ref{tab:esm_md_results} provides the evaluation of ESM\textsc{Flow} on MD ensembles. In Table~\ref{tab:md_results_app}, we report the performance of Alpha\textsc{Flow}-MD with ablated training procedures, and the comparison of Alpha\textsc{Flow} with normal mode analysis conducted on the PDB structure. In Figures~\ref{fig:md_rmsf}--\ref{fig:md_cyptic_residues}, we provide additional visualizations for the RMSF, transient contacts, weak contacts, and solvent exposure analyses of Alpha\textsc{Flow}-MD ensembles. Finally, In Figure~\ref{fig:md_convergence_app}, we provide additional convergence results for Alpha\textsc{Flow}-MD+Templates vs replicate MD simulations.
\vfill

\begin{table}[h!]
    \centering
    \caption{\textbf{Evaluation of ESM\textsc{Flow} on MD ensembles}}
    \label{tab:esm_md_results}
    \begin{small}
    \begin{tabular}{clccc|cc}
    \toprule
    & & \multicolumn{2}{c}{ESM\textsc{Flow}-MD} & & \multicolumn{2}{c}{EFMD+Templates}\\ 
    \cmidrule(lr){3-4} \cmidrule(lr){6-7}
    & & Full & Distilled & ESMFold & Full & Distilled\\
    \midrule
    \multirow{5}{*}{\makecell{Predicting\\flexibility}}
    & Pairwise RMSD ($=$2.90) & 3.25 & 2.76	& 0.00 & 2.00 & 1.42 \\
    & Pairwise RMSD $r$ $\uparrow$ & 0.19 & 0.19 & --- & 0.85 & 0.80\\
    & All-atom RMSF ($=$1.70) & 2.16 & 2.12 & 0.00 & 1.07 & 0.80 \\
    & Global RMSF $r$ $\uparrow$ & 0.31 & 0.33 & --- & 0.84 & 0.79\\
    & Per-target RMSF $r$ $\uparrow$ & 0.76	& 0.74 & --- & 0.90 & 0.87\\
    \midrule
    \multirow{6}{*}{\makecell{Distributional\\accuracy}} 
    & Root mean $\mathcal{W}_2$-dist. $\downarrow$ & 3.60 &	4.23 &	4.60 & 	2.17 &	2.27\\
    & $\hookrightarrow$ Translation contrib. $\downarrow$ & 3.13 &	3.75 &	3.65 &	1.66	& 1.70\\
    & $\hookrightarrow$ Variance contrib. $\downarrow$ & 1.74 &	1.90 &	2.50 &	1.07	& 1.35 \\
    & MD PCA $\mathcal{W}_2$-dist. $\downarrow$ & 1.51 & 1.87 & 1.69 & 1.44 & 1.48\\
    & Joint PCA $\mathcal{W}_2$-dist. $\downarrow$  & 3.19 & 3.79 & 3.87 & 1.70 & 1.81\\
    & \% PC-sim $>0.5$ $\uparrow$ & 26	& 33	& ---	& 49 & 	40 \\
    \midrule
    \multirow{4}{*}{\makecell{Ensemble\\observables}}  
    & Weak contacts $J$ $\uparrow$ & 0.55& 	0.48	& 0.22 & 	0.59 & 0.48\\
    & Transient contacts $J$ $\uparrow$ & 0.34	&0.30&	0.15&	0.47&	0.41\\
    & Exposed residue $J$ $\uparrow$ & 0.49	& 0.43	&0.28	&0.50&	0.44\\
    & Exposure MI matrix $\rho$ $\uparrow$ & 0.20	& 0.16 & --- & 0.22 & 0.16\\
    \bottomrule
    \end{tabular}
    \end{small}
\end{table}

\vfill

\begin{table}[h!]
    \centering
    \caption{\textbf{Ablations and normal mode analysis on MD ensembles}. The ablations verify the importance of the two-step training procedure for Alpha\textsc{Flow}+MD. Normal mode analysis (NMA) often fails to outperform baseline Alpha\textsc{Flow}+MD despite having access to the ground truth PDB structure, and significantly underperforms Alpha\textsc{Flow}+MD+Templates when it is provided the same PDB template structure. $^\star$Note that NMA results for RMSF and RMWD are C$\alpha$-only rather than all-atom, which likely overestimates the performance. GNM: Gaussian Network Model; ANM: Anisotropic Network Model.}
    \label{tab:md_results_app}
    \begin{small}
    \begin{tabular}{clccc|ccc}
    \toprule
    & & & \multicolumn{2}{c}{Ablations} & & \multicolumn{2}{c}{NMA}\\ 
    \cmidrule(lr){4-5} \cmidrule(lr){7-8}
    & & Baseline & \makecell{No ATLAS\\finetuning} & \makecell{No PDB\\pretraining} & \makecell{AFMD\\+Templates} & GNM & ANM \\
    \midrule
    \multirow{5}{*}{\makecell{Predicting\\flexibility}}
    & Pairwise RMSD ($=$2.90)   & \textbf{2.89} & 2.41 & 3.04 & 2.18 & 1.85 &  2.36 \\
    & Pairwise RMSD $r$ $\uparrow$ & \textbf{0.48} & 0.34 & 0.29 & \textbf{0.94} & 0.71 & 0.65 \\
    & All-atom RMSF ($=$1.70) & \textbf{1.68} & 1.25 & 1.81 & 1.31 & 1.22$^\star$ & 1.35$^\star$ \\
    & Global RMSF $r$ $\uparrow$ & \textbf{0.60} & 0.48 & 0.45 & \textbf{0.91} & 0.64$^\star$ & 0.55$^\star$ \\
    & Per-target RMSF $r$ $\uparrow$  &       \textbf{0.85} & 0.82 & 0.83 &  \textbf{0.90} & 0.72$^\star$ & 0.76$^\star$ \\
    \midrule
    \multirow{6}{*}{\makecell{Distributional\\accuracy}} 
    & Root mean $\mathcal{W}_2$-dist. $\downarrow$ &       \textbf{2.61} & 2.96 & 3.11 & \textbf{1.95} & 2.47$^\star$ & 2.54$^\star$ \\
    & $\hookrightarrow$ Translation contrib. $\downarrow$  &       \textbf{2.28} & 2.52 & 2.71 & \textbf{1.64} & 2.07$^\star$ & 2.09$^\star$\\
    & $\hookrightarrow$ Variance contrib. $\downarrow$ &       \textbf{1.30} & 1.36 & 1.44 & \textbf{1.01} & 1.33$^\star$ & 1.27$^\star$ \\
    & MD PCA $\mathcal{W}_2$-dist. $\downarrow$ &       \textbf{1.52} & 1.64 & 1.59 & \textbf{1.25} & 1.84 & 1.73 \\
    & Joint PCA $\mathcal{W}_2$-dist. $\downarrow$  &       \textbf{2.25} & 2.60 & 2.67 & \textbf{1.58} & 2.44 & 2.35\\
    & \% PC-sim $>0.5$ $\uparrow$  &       \textbf{44} & 35 & 38 & \textbf{44} & 13 & 23 \\
    \midrule
    \multirow{4}{*}{\makecell{Ensemble\\observables}}  
    & Weak contacts $J$ $\uparrow$ &       \textbf{0.62} & 0.48 & 0.60 & \textbf{0.62} & 0.45 & 0.40 \\
    & Transient contacts $J$ $\uparrow$ &       \textbf{0.41} & 0.36 & 0.39 & \textbf{0.47} & 0.25 & 0.25\\
    & Exposed residue $J$ $\uparrow$ &       \textbf{0.50} & 0.40 & \textbf{0.50} & 0.50 & --- & --- \\
    & Exposure MI matrix $\rho$ $\uparrow$ &       \textbf{0.25} & 0.18 & \textbf{0.25} & 0.25 & --- & --- \\
    \bottomrule
    \end{tabular}
    \end{small}
\end{table}

\vfill
\clearpage

\begin{figure}
    \centering
    \includegraphics[width=0.49\textwidth]{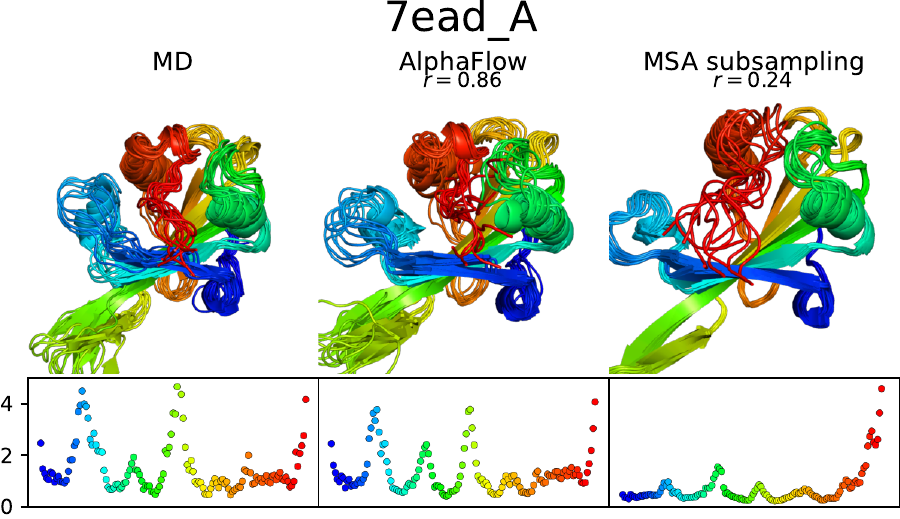}
    \includegraphics[width=0.49\textwidth]{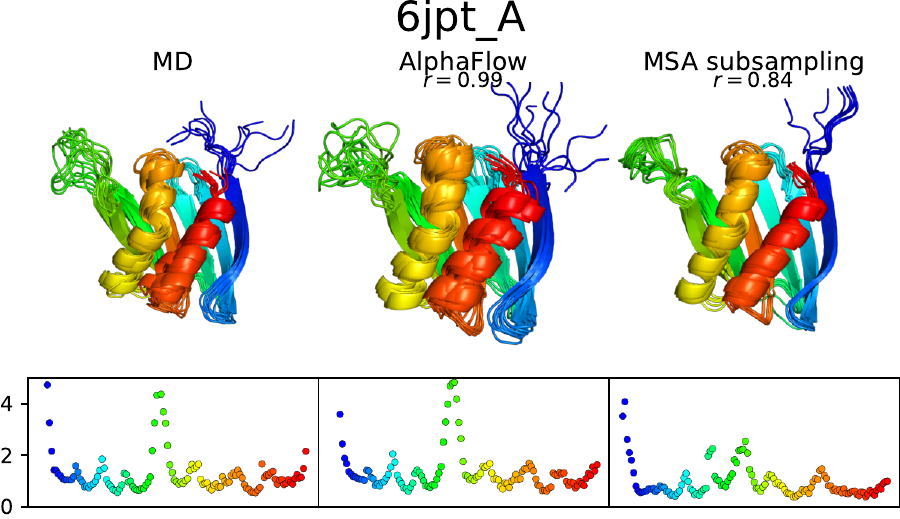} \\
    \vspace{12pt}
    \includegraphics[width=0.49\textwidth]{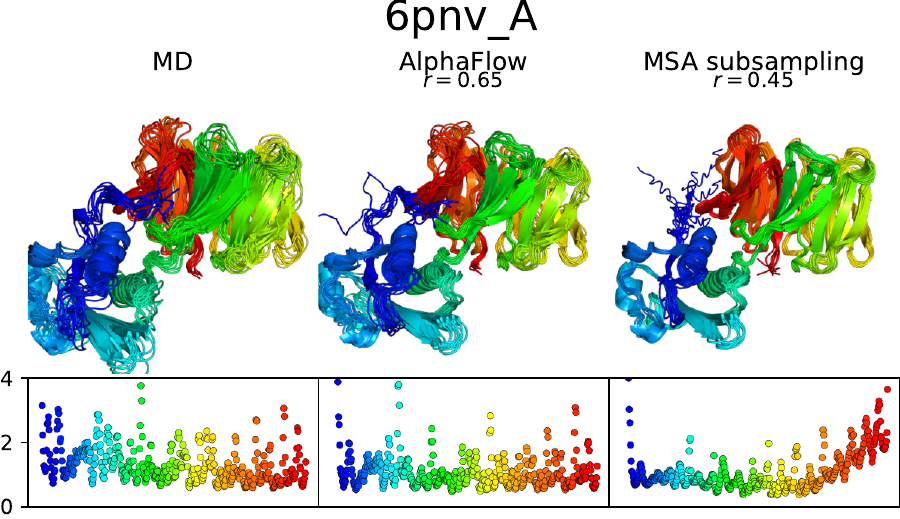} 
    \includegraphics[width=0.49\textwidth]{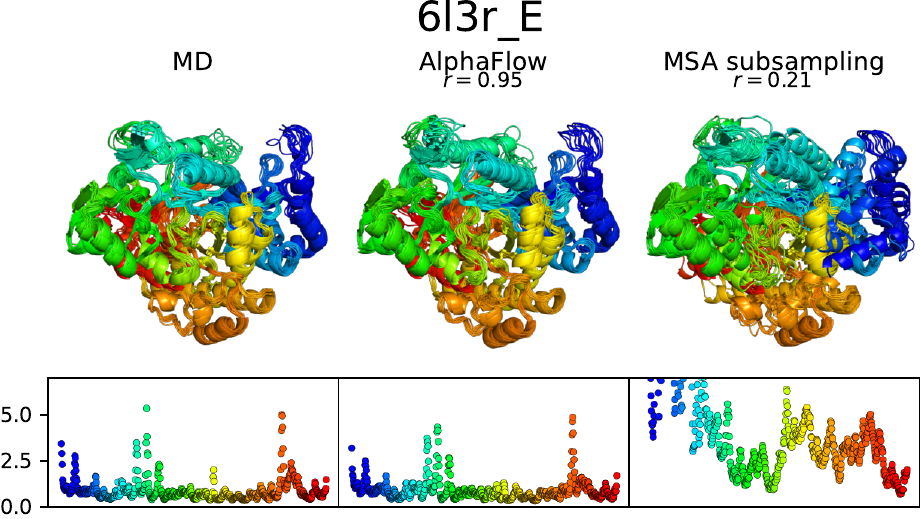}
    \caption{\textbf{Visualization of ensembles and their RMSF plots}. For each PDB ID, 10 samples from the MD, Alpha\textsc{Flow}, and MSA subsampling (depth 48) ensembles are shown, with RMSF by residue index in insets. For the latter two, the Pearson correlation coefficient ($r$) with the MD RMSF is reported.}
    \label{fig:md_rmsf}
\end{figure}

\clearpage

\begin{figure}
    \centering
    \includegraphics[width=\textwidth]{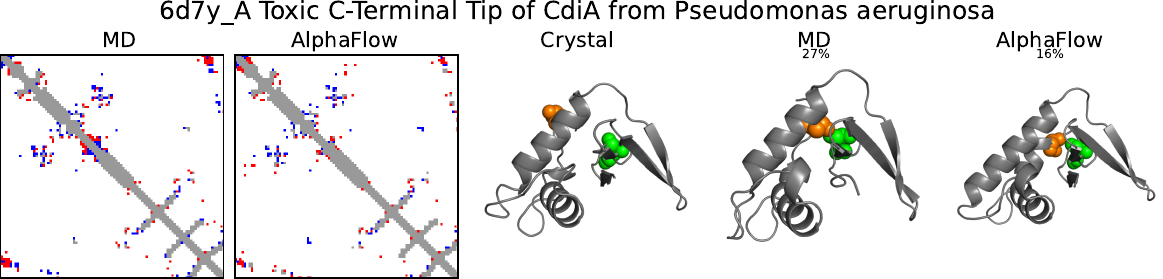}\\
    \vspace{12pt}
    \includegraphics[width=\textwidth]{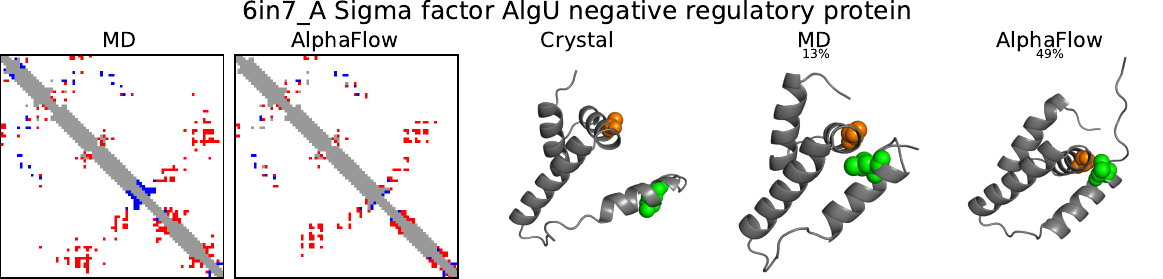}\\
    \vspace{12pt}
    \includegraphics[width=\textwidth]{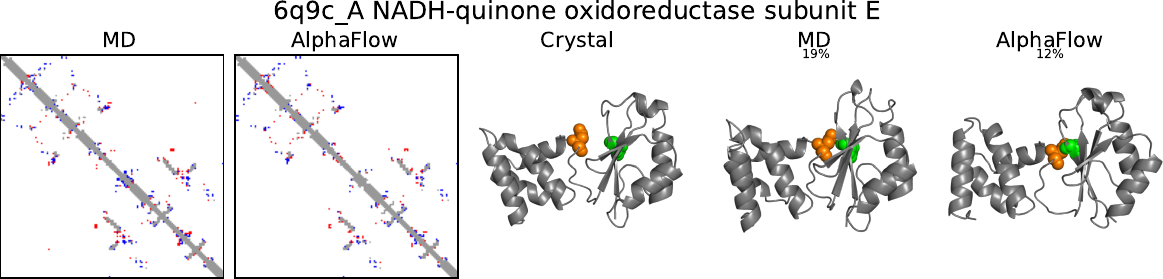}\\
    \vspace{12pt}
    \includegraphics[width=\textwidth]{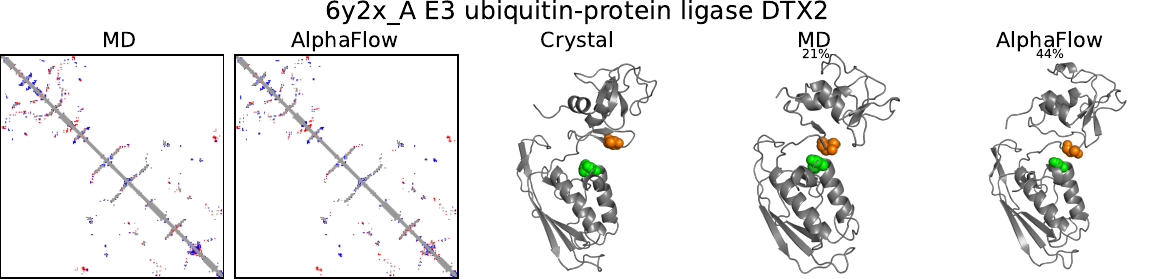}\\
    \vspace{12pt}
    \caption{\textbf{Visualization of transient contacts.} For each PDB ID, the contact maps from MD simulation and Alpha\textsc{Flow} are shown, with normal contacts in \textcolor{gray}{gray}, weak contacts in \textcolor{blue}{blue}, and transient contacts in \textcolor{red}{red}. Among the the transient contacts correctly identified by Alpha\textsc{Flow}, one is selected for visualization: the two residues are highlighted in the crystal structure (\emph{left}), a frame from the MD simulation (\emph{middle}) where they are in contact, and an Alpha\textsc{Flow} sample where they are in contact. The probability of occurence in each ensemble is shown.}
    \label{fig:md_transient_contacts}
\end{figure}

\clearpage

\begin{figure}
    \centering
    \includegraphics[width=\textwidth]{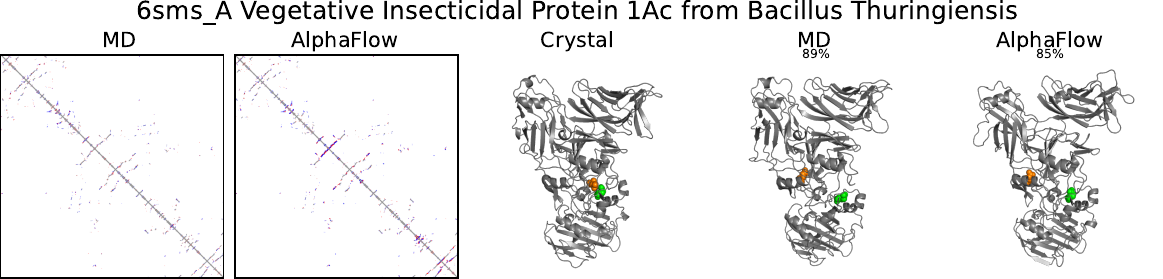}\\
    \vspace{12pt}
    \includegraphics[width=\textwidth]{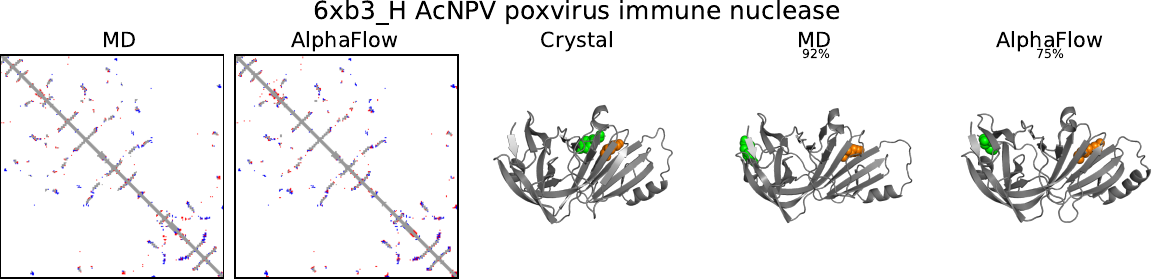}\\
    \vspace{12pt}
    \includegraphics[width=\textwidth]{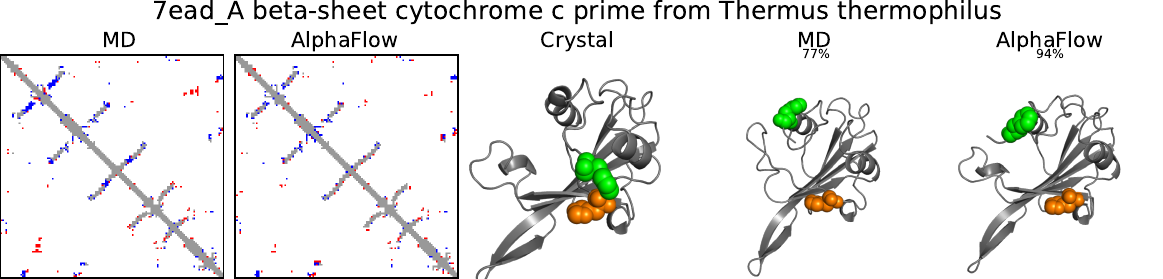}\\
    \vspace{12pt}
    \includegraphics[width=\textwidth]{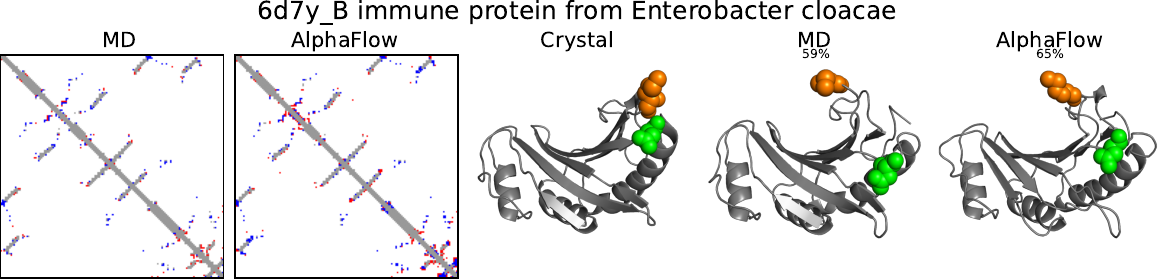}\\
    \vspace{12pt}
    \caption{\textbf{Visualization of weak contacts.} For each PDB ID, the contact maps from MD simulation and Alpha\textsc{Flow} are shown, with normal contacts in \textcolor{gray}{gray}, weak contacts in \textcolor{blue}{blue}, and transient contacts in \textcolor{red}{red}. Among the the weak contacts correctly identified by Alpha\textsc{Flow}, one is selected for visualization: the two residues are highlighted in the crystal structure (\emph{left}), a frame from the MD simulation (\emph{middle}) where they are not in contact, and an Alpha\textsc{Flow} sample where they are not in contact. The probability of occurence in each ensemble is shown.}
    \label{fig:md_weak_contacts}
\end{figure}

\clearpage

\begin{figure}
    \centering
    \includegraphics[width=0.9\textwidth]{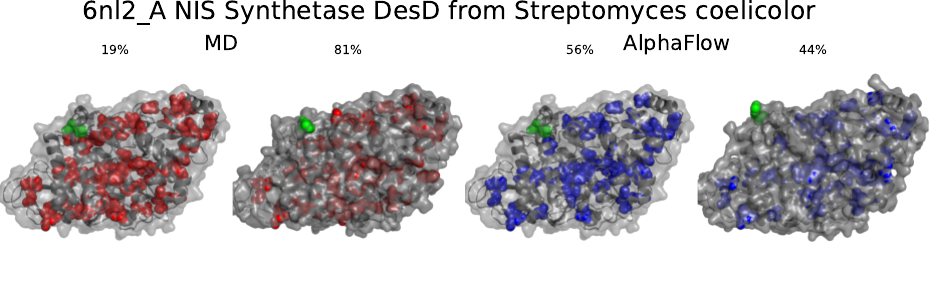}\\
    \vspace{12pt}
    \includegraphics[width=0.9\textwidth]{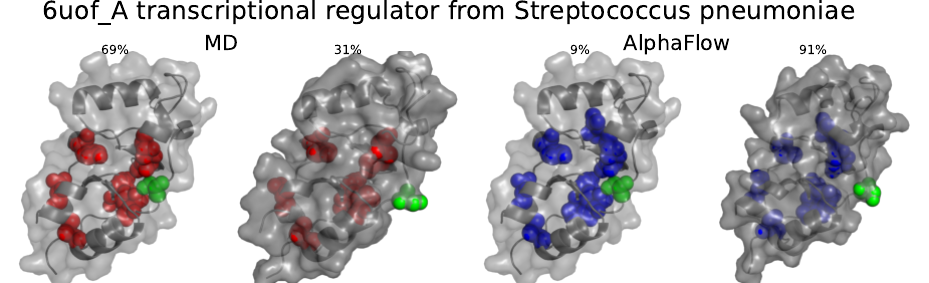}\\
    \vspace{12pt}
    \includegraphics[width=0.9\textwidth]{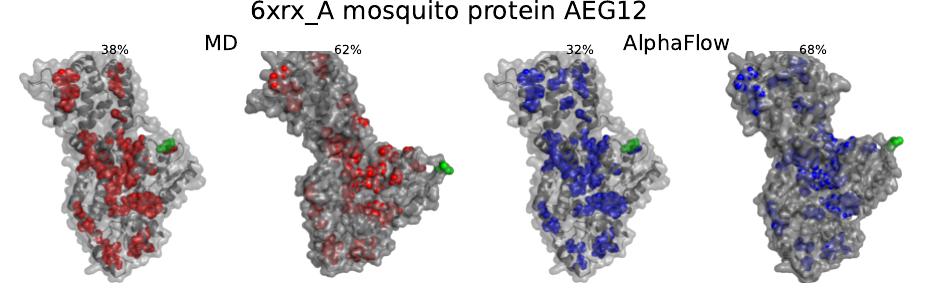}\\
    \vspace{12pt}
    \includegraphics[width=0.9\textwidth]{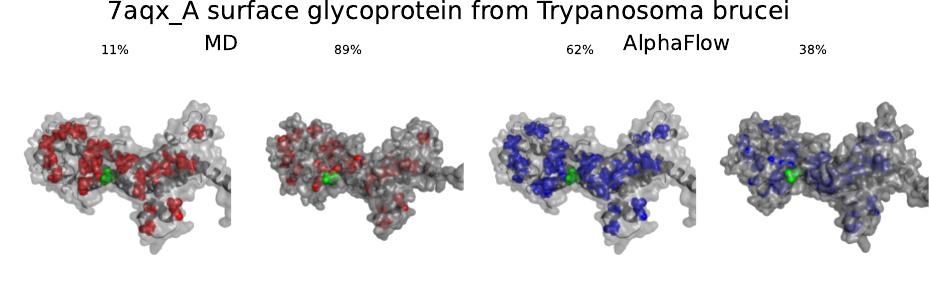}\\ 
    \vspace{12pt}
    \caption{\textbf{Visualization of cryptic exposed residues.} For each PDB ID, in the \emph{left} pair of structures, the set of true cryptic exposed residues (from MD) is colored \textcolor{red}{red}; in the \emph{right} pair the set identified from Alpha\textsc{Flow} ensembles is colored \textcolor{blue}{blue}. A common identified residue is selected and highlighted in \textcolor{green}{green}. For each pair, the \emph{left} structures shows the residues buried in the crystal structure whereas the \emph{right} structure shows a frame (or sample) where the highlighted residue is exposed to the solvent. The probability of occurence in each ensemble is shown.}
    \label{fig:md_cyptic_residues}
\end{figure}

\clearpage

\begin{figure}
    \centering
    \includegraphics[width=0.75\textwidth]{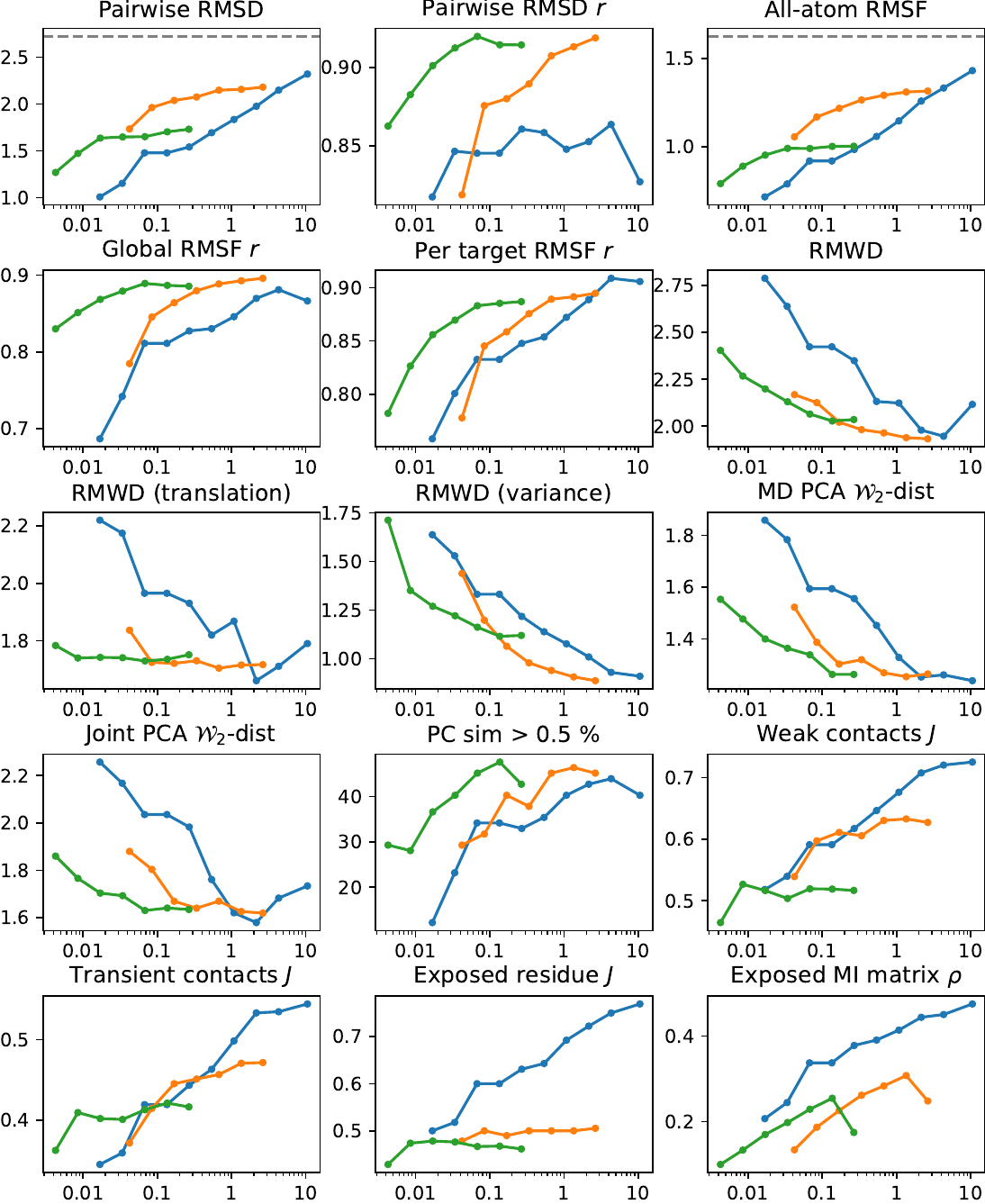}
    \caption{\textbf{Efficiency of Alpha\textsc{Flow} vs replicate MD simulations.} Alpha\textsc{Flow} (with templates) with varying number of samples in \textcolor[HTML]{ff7f0e}{orange}; Alpha\textsc{Flow} distilled into a single forward pass in \textcolor[HTML]{2ca02c}{green}; MD with varying trajectory lengths in \textcolor[HTML]{1f77b4}{blue}. For Pairwise RMSD and RMSF, the values from the reference MD (i.e., pooling the remaining two replicates) are shown as horizontal dashed lines. The $x$-axis reports runtime in GPU-hrs averaged over targets. For MD, the average runtime is 6.3 mins / ns; for Alpha\textsc{Flow}, the average runtime per sample is 38 s without distillation and 3.8 s with distillation. See Appendix~\ref{app:evaluation} for further benchmarking details.}
    \label{fig:md_convergence_app}
\end{figure}

\clearpage


\end{document}